\documentclass[11pt]{article}
\usepackage[font={small,it}]{caption}
\usepackage{subcaption}
\usepackage[usenames,dvipsnames]{xcolor}
\usepackage[compact]{titlesec}

\usepackage{typearea}
\typearea{14}


\usepackage[linktocpage=true,
	pagebackref=true,colorlinks,
	linkcolor=BrickRed,citecolor=blue,hypertexnames=false]
	{hyperref}
	\usepackage{latexsym,graphicx,epsfig,color}
\usepackage{amsfonts,amssymb,amsmath,amsthm,amstext}
\usepackage{libertine}
\usepackage[libertine]{newtxmath}
\usepackage{enumitem}
\setitemize{itemsep=2pt,topsep=0pt,parsep=0pt}
\usepackage{url,setspace}
\usepackage{multirow}
\usepackage{rotating}
\usepackage{makeidx}
\usepackage{accents}
\usepackage{xspace}
\usepackage{algorithm}
\usepackage[noend]{algpseudocode}
\usepackage{bm}
\usepackage{changepage} 	
\usepackage{thmtools,thm-restate}
\usepackage{cleveref}
\usepackage{nicefrac}

\Crefname{algocf}{Algorithm}{Algorithms}
\crefname{algocfline}{line}{lines}
\Crefname{invariant}{Invariant}{Invariants}
\Crefname{claim}{Claim}{Claims}
\Crefname{subclaim}{Subclaim}{Subclaims}

\newcommand{\IGNORE}[1]{}
\allowdisplaybreaks

\newcommand{\snote}[1]{{\color{red}\textsl{\small[#1]}\marginpar{\tiny\textsc{SS-Comment}}}}
\newcommand{\vgnote}[1]{{\color{blue}\textsl{\small[#1]}\marginpar{\tiny\textsc{VG-Comment}}}}

\makeindex

\newtheorem{theorem}{Theorem}[section]
\newtheorem{claim}[theorem]{Claim}

\newtheorem{lemma}[theorem]{Lemma}
\newtheorem{corollary}[theorem]{Corollary}

\theoremstyle{definition}

\newtheorem{remark}[theorem]{Remark}

\newtheorem{open}[theorem]{Open Problem}
\newtheorem{definition}[theorem]{Definition}

\usepackage{thmtools,thm-restate} 

\newcommand{\disc}{{\mathsf{disc}}}

\newcommand{\currG}[1]{G_{\mathsf{curr}}(#1)}
\newcommand{\levelG}[2]{G_{#1}(#2)}
\newcommand{\expG}[3]{G_{#1,#2}(#3)}

\newcommand{\dir}[1]{{\vec{#1}}}
\newcommand{\dirp}[1]{{{\widetilde #1}}}

\newcommand{\updatevec}{{\textproc{UpdateVector}}\xspace}
\newcommand{\yold}{y^{old}}
\newcommand{\ynew}{y^{new}}

\newcommand{\calT}{\mathcal{T}}
\newcommand{\DBG}{{\textproc{DBG}}\xspace}
\newcommand{\DBGUpdate}{{\textproc{DBGUpdate}}\xspace}

\def\ist{{i^\star}}

\def\eps {\epsilon}

\newcommand{\one}{\mathbf{1}\xspace}
\newcommand{\poly}{\operatorname{poly}}

\newcommand{\sse}{\subseteq}
\newcommand{\nf}{\nicefrac}
\renewcommand{\emptyset}{\varnothing}

\newcommand{\bargrin}{\textsf{B\'{a}r\'{a}ny-Grinberg}\xspace}

\newcommand{\ovb}{\textsc{Online Vector Balancing}\xspace}
\newcommand{\fdvb}{\textsc{Fully-Dynamic Vector Balancing}\xspace}

\newcommand{\vol}{\mathsf{vol}}
\newcommand{\local}{\textproc{Local-Search}\xspace}

\newcommand{\pathlocal}{\textproc{Path-Local-Search}\xspace}
\newcommand{\Prune}{\textproc{Prune}\xspace}

\newcommand{\e}{\varepsilon}
\newcommand{\R}{\mathbb{R}}

\newcommand{\DR}{\textproc{Dyadic Resigning}\xspace}
\newcommand{\cA}{\ensuremath{\mathcal{A}}}

\newcommand{\instant}{\text{moment}\xspace}
\newcommand{\instants}{\text{moments}\xspace}

\title{Online Discrepancy with Recourse for Vectors and Graphs}

\author{
	Anupam Gupta\thanks{ (anupamg@cs.cmu.edu) 	Computer Science Department, 	Carnegie Mellon University. Supported in part by NSF awards CCF-1907820, CCF1955785, and CCF-2006953.}
	\and Vijaykrishna Gurunathan \thanks{(krishnavijay1999@gmail.com) Computer Science Department, Stanford University}
	\and Ravishankar Krishnaswamy \thanks{	(rakri@microsoft.com) 	Microsoft Research.	}
		\and Amit Kumar \thanks{ (amitk@cse.iitd.ac.in) 		Department of Computer Science and Engineering, 		Indian Institute of Technology Delhi.		}
		\and Sahil Singla \thanks{        (ssingla@gatech.edu)        School of Computer Science,         Georgia Tech.         }
}

\begin{document}



\date{}

\maketitle



\begin{abstract}
  The vector-balancing problem is a fundamental problem in
  discrepancy theory: given $T$ vectors in $[-1,1]^n$, find a signing
  $\sigma(a) \in \{\pm 1\}$ of each vector $a$ to minimize the
  discrepancy $\| \sum_{a} \sigma(a) \cdot a \|_{\infty}$. This problem has
  been extensively studied in the static/offline setting. In this paper we
  initiate its study in the fully-dynamic setting with \emph{recourse}: the
  algorithm sees a stream of $T$ insertions and deletions of vectors,
  and at each time must maintain a low-discrepancy signing, while also
  minimizing the amortized recourse (the number of times any vector
  changes its sign) per update.

  For general vectors, we show algorithms which almost match Spencer's
  $O(\sqrt{n})$ offline discrepancy bound, with
  ${O}(n\poly\!\log T)$ amortized recourse per update. The crucial idea
  behind our algorithm is to compute a basic feasible solution to the
  linear relaxation in a distributed and recursive manner, which helps find a low-discrepancy signing. We bound the recourse
  using the distributed computation of the basic solution, and argue that only
  a small part of the instance needs to be re-computed at each update.

  Since vector balancing has also been greatly studied for sparse vectors, 
  we then give algorithms for \emph{low-discrepancy edge orientation},
  where we dynamically maintain signings for $2$-sparse vectors in an
  $n$-dimensional space. Alternatively, this can be seen as orienting
  a dynamic set of edges of an $n$-vertex graph to minimize the
  discrepancy, i.e., the absolute difference between in- and
  out-degrees at any vertex. We present a deterministic algorithm with
  $O(\poly\!\log n)$ discrepancy and $O(\poly\!\log n)$ amortized
  recourse. The core ideas are to dynamically maintain an
  expander-decomposition with low recourse (using a very simple
  approach), and then to show that, as the expanders change over time,
  a natural local-search algorithm converges quickly (i.e., with low
  recourse) to a low-discrepancy solution. We also  give strong lower
  bounds (with some matching upper bounds) for local-search  discrepancy minimization algorithms for
   vector balancing and edge orientation. 
\end{abstract}







\section{Introduction}
\label{sec:introduction}

In the 
\ovb problem introduced by Spencer~\cite{Spencer77}, vectors
$a_1, a_2, \ldots, a_T \in [-1,1]^n$ arrive online, and the algorithm
irrevocably assigns a sign $\sigma(a_t)$ immediately upon seeing $a_t$,
with the goal of minimizing the \emph{discrepancy of the signed sum},
i.e., $\| \sum_t \sigma(a_t)\cdot a_t \|_{\infty}$. Following a sequence of works~\cite{BS20,BJSS20,BJMSS-SODA21}, the state-of-the-art bounds for this problem is an elegant randomized algorithm that maintains a
discrepancy of $O(\sqrt{n} \log (nT))$~\cite{ALS-STOC21}. Their result assumes an
\emph{oblivious} adversary, so that the choice of arriving vectors does not
depend on the internal state of the algorithm. 
Indeed, if we allow \emph{adaptive}
adversaries then every online algorithm incurs $\Omega(\sqrt{T})$ discrepancy~\cite{Spencer77}.
We think of $T \gg n$, so $\Omega(\sqrt{T})$ is much larger than $O(\sqrt{n} \log (nT))$.

 
We initiate the study of  \fdvb, where vectors can both arrive
or depart at each time step, and the algorithm must always maintain a
low-discrepancy signing of the vectors present in the system at all
times. Since it is easy to construct examples where no algorithm can
guarantee non-trivial discrepancy bounds if it is forced to commit to
the sign of a vector upon arrival, we study the problem where the
algorithm can \emph{re-sign the vectors from time to time}. Indeed,
many real-world applications that motivate such discrepancy-based
methods (such as in fair allocations, sparsification routines, etc.)
have a fully-dynamic flavor to them, with the corresponding inputs
being dynamic in nature due to both insertions and deletions. 


\textbf{Problem} (\fdvb). \textit{
  We start with an empty collection of \emph{active} vectors $A(0)$. At each time/update
  $t \in [T]$, an adaptive adversary either inserts a new vector
  $a_t \in [-1,1]^n$, i.e., $A(t) = A(t-1) \cup \{a_t\}$, or removes
  an existing vector $a \in A(t-1)$, i.e.,
  $A(t) = A(t-1) \setminus \{a\}$. The goal is to maintain
  \emph{signings} $\sigma_t: A(t) \rightarrow \{\pm 1\}$ to minimize
  the norm $\| \sum_{a \in A(t)} \sigma_t(a) \cdot a \|_{\infty}$. The
  algorithm can reassign the sign of a vector $a$ (i.e., set
  $\sigma_t(a) \neq \sigma_{t-1}(a)$), and the total \emph{recourse} is the sum total of the reassignments.
  }

Two trivial solutions exist: (a)~recomputing low-discrepancy signings
on the active set  of vectors after every update operation incurs optimal offline discrepancy guarantees with a recourse of $\Theta(T)$ per
update, and (b) an independent and uniformly random signing of every
new vector 
maintains  at any time $t$ a signing of discrepancy
$\Theta(\sqrt{T \log n})$ w.h.p., while performing no recourse
whatsoever. 
Since $T \gg n$, this is much larger than the
optimal offline discrepancy bounds of $O(\sqrt{n})$ for any collection of $T$ vectors in 
$[-1,1]^n$~\cite{Spencer85,Bansal-FOCS10,Lovett-Meka-SICOMP15}.
We ask: \emph{can we get near-optimal\footnote{In this paper, we use ``near-optimal'' to mean optimal up to poly-logarithmic factors.} discrepancy bounds with a small amount of
recourse?}

\subsection{Our Results and Techniques}

\paragraph{Fully-Dynamic Vector Balancing.}


Our first main contribution is the design of an  algorithm which maintains low-discrepancy
signings for the fully-dynamic problem that nearly matches the offline discrepancy
bounds while giving an amortized recourse that is only logarithmic in the
sequence length $T$.

\begin{theorem}[Fully-Dynamic: General Vector Balancing]
  \label{thm:main1}
  There is an efficient algorithm for \fdvb with update 
  vectors in $[-1,1]^n$ which maintains signings
  $\sigma_t(\cdot)$ with discrepancy $O(\sqrt{n})$ and an amortized recourse of $O(n \log T)$ per update, even against adaptive   adversaries.   For Komlos' setting, i.e., if all the updates vectors have $\ell_2$ length at most $1$ (instead of $\ell_\infty$ length),  the algorithm achieves discrepancy $O(\sqrt{\log(n)})$ with an amortized  recourse of $O(n \log T)$ per update.
\end{theorem}

Since in this theorem we are competitive against \emph{adaptive adversaries}, it illustrates the power of recourse:  in the absence of recourse, we get $\Omega(\sqrt{T})$ lower bounds 
on the discrepancy even for arrival-only sequences of $2$-dimensional vectors. This is because the adversary can always make the next vector to be
orthogonal to the current signed sum.

At a very high level, our algorithm divides the instance into many
parts of size $O(n)$, obtains a good partial signing for each part
(such that all but $n$ vectors are signed), and recurses on the
residual instance. The algorithm imposes a tree-like hierarchy on
these parts, so that it can easily adapt to inserts or deletes with
bounded recourse by only re-running the computations on the part
suffering the insertion/deletion, and on any internal node on the
corresponding root-leaf path from that part to the root. If we are not
careful, the discrepancy of the overall vector can be proportional to
the number of parts, since we could accrue error in each
part. However, we use linear algebraic ideas inspired
by~\cite{BaranyGrinberg81} to \emph{couple all the parts}, thereby
always ensuring that the sum of the partial signings across all nodes
of the tree (except the root) is zero. 


\paragraph{Fully-Dynamic Edge Orientation/Carpooling and Local Search.}

Given the general result above,  next we focus on the special case of
orienting edges of a graph to minimize the maximum imbalance between
the in- and out-degrees. 
Fagin and Williams~\cite{Fagin83}
posed the \emph{carpooling problem},
which corresponds to vector balancing with vectors of the form
$(0, \ldots, 0, 1, 0, \ldots, 0, -1, 0, \ldots) \in \R^n$, and the
graph discrepancy objective is precisely the $\| \cdot \|_{\infty}$ of
the signed sum of vectors.
\cite{Fagin83,Ajtai98} use this problem to model fairness in
scheduling, where edges represent shared commitments (such as
carpooling), orientations give primary and secondary partners of the
commitment (e.g., driver and co-driver), and hence the discrepancy
measures fairness for individuals, in terms of how many commitments
he/she is the primary partner for, relative to the total number of
commitments he/she is a part of.

Somewhat surprisingly, \cite{Ajtai98} showed that any algorithm must
suffer $\Omega(n)$ discrepancy on some worse-case adaptive sequence of
edge arrivals. On the other hand, for an oblivious sequence of edge
arrivals, it is easy to maintain orientations with
$O(\sqrt{n \log n})$ discrepancy by simply orienting edges randomly
(while always orienting repeated parallel edges $(u,v)$ oppositely).
To mitigate such strong lower bounds, \cite{Ajtai98} and recently
Gupta et al.~\cite{GuptaKKS20} study a stochastic version of the
problem where the arriving edges are sampled from a known distribution:
they design algorithms to maintain $\poly\!\log(n,T)$-discrepancy. The
recent algorithm of Alweiss et al.~\cite{ALS-STOC21} also extends to
this special case giving $O(\log (nT))$ discrepancy bounds for any
oblivious sequence of edge arrivals, not just stochastic ones. None of these prior algorithms extend to a
fully-dynamic input consisting of both insertions and deletions. 
Moreover,  \Cref{thm:main1} guarantees near-optimal discrepancy  only with $O(n \log T)$ amortized recourse. 
In this paper, we give \emph{deterministic} near-optimal discrepancy algorithms with near-optimal amortized recourse. 

\begin{theorem}[Fully-Dynamic Edge Orientation]
  \label{thm:main2}
  There is an efficient deterministic algorithm that maintains an
  orientation of $\poly\!\log n$ discrepancy while performing an
  amortized recourse of $\poly\!\log n$ per update.
\end{theorem}

Since this algorithm   is deterministic, 
  the guarantees also hold against
adaptive adversaries: there are $\Omega(n)$
discrepancy  bounds for no-recourse algorithms against such adversaries, even
for the setting of only arrivals.

At a high level, our algorithm can be seen as a composition of two
modules. Firstly, we consider a simple local-search procedure, which
flips an edge from $u \rightarrow v$ to $v \rightarrow u$ if the
current discrepancy of $v$ exceeds that of $u$ by more than
$2$. Clearly, this reduces the discrepancy of the maximum of these two
vertices. Our crucial observation is that  this process always maintains
low-discrepancy signings when the graph is an \emph{expander}. We find
this interesting, since we can show that there are bad local optima with
$\poly(n)$ discrepancy for general graphs.
Secondly, we show how to dynamically maintain a partitioning of the
edge set of an arbitrary graph $G$ into a disjoint collection of
expanders $G_1, G_2, \ldots, G_{\ell}$ with each vertex appearing in
at most $\poly\!\log n$ many expanders, such that the \emph{amortized
  number of changes} to $G_1$, $G_2$, $\ldots, G_\ell$,  per
update to $G$ is bounded. (This expander decomposition can be viewed
as a ``preconditioning'' step.)
We build on ideas recently developed for dynamic graph
algorithms~\cite{saranurak2019expander,bernstein2020fullydynamic}: our challenge is to show
that dynamic expander decomposition can be done along with
local search on the individual expanders, and specifically to control
the potential functions that guide our proofs.

Indeed, using the above two modules to obtain~\Cref{thm:main2}
requires new ideas.
When an update  (insertion or deletion) occurs
to $G$, we first modify our expander decomposition,
and re-run local search \emph{starting from the prior local optima} in
each expander. While this ensures good discrepancy bounds, it could lead to  many local search moves.
In order to bound the latter quantity, 
our  idea
is to use a potential
function in each expander to bound the recourse, such that each step
of local search  always decreases the potential by at least a constant. This
is somewhat delicate: a single update in $G$ can change any particular
expander $G_i$ by a lot 
(even though the amortized
recourse is bounded), and hence the single-step potential change can
be huge. We show how to maintain some Lipschitzness properties for our
potential function under inserts and deletes, which gives us the final bounds
of $\poly\!\log n$ on both the discrepancy and the
recourse. 

Along the way, we also develop a better understanding of the strengths
and limitations of local search as a technique for discrepancy minimization
problems, both for  graphs and  for general vectors.


\begin{theorem}[Informal: Discrepancy of Local Optima]
  \label{thm:main3}
  For edge orientation in expanders, any locally optimal solution for
  local search using the simple potential
  $\Phi = \sum_{v \in V} disc(v)^2$ has discrepancy  $O(\log
  n)$. For arbitrary graphs, however, the discrepancy can be as bad as
  $\Omega(n^{1/3})$. For general vectors in $\{\pm 1\}^n$ (and in
  $[-1,1]^n$), the local search bound using the simple potential   $\Phi = \sum_{i \in [n]} \big(\sum_t \sigma(a_t)\cdot a_t(i) \big)^2$
  deteriorates to $\Omega(2^n)$ (and to
  $\Omega(\sqrt{T})$). 
\end{theorem} 

\medskip \noindent {\bf Signing $s$-Sparse Vectors for Online Arrivals.} Finally,
we consider the problem with $s$-sparse vectors, which interpolates
between the graphical case of $s=2$ and the general case. 
In the offline setting, the classical linear-algebraic algorithm of
Beck and Fiala~\cite{BeckFiala-DAM81} constructs a signing with
disrepancy $2s-1$ (independent of $n$ and $T$). Subsequent works by
Banaszczyk~\cite{Banaszczyk-Journal98} and Bansal, Dadush, and Garg~\cite{BansalDG-FOCS16} develop techniques to get 
discrepancy $O(\sqrt{s \log n})$, and a long-standing question in
discrepancy theory is to improve this bound to $O(\sqrt{s})$.
Here we study the \ovb problem only with arrivals. In this setting, the algorithm of \cite{ALS-STOC21} maintains signings
of discrepancy $O(\sqrt{s} \log (nT))$ without recourse, but against
an \emph{oblivious} adversary. In \Cref{sec:insert-only} we give a generic reduction that
can maintain near-optimal discrepancy  for \ovb
against adaptive adversaries, with small recourse.

\begin{theorem}[Arrivals Only: Online Vector Balancing with Recourse]
  \label{thm:main4}
  There is an efficient algorithm for \ovb with $s$-sparse vectors
  that achieves $O(\sqrt{s \log n} \log T)$ discrepancy and $O(\log T)$
  amortized recourse per update against an adaptive adversary.
\end{theorem}

\subsection{Further Related Work}
\label{sec:related-work}

Discrepancy theory is a rich and vibrant area of research~\cite{Chazelle-Book01,Matousek-Book09}. While some
initial works~\cite{Spencer77,Barany79} focused on the online
discrepancy problem, the majority of research dealt with the offline
setting, where the $T$ vectors are given upfront. Near-optimal
results are known for settings such as discrete set
systems~\cite{Spencer85,Bansal-FOCS10,Lovett-Meka-SICOMP15} (i.e., vectors in $\{0,1\}^n$), \emph{sparse}
set systems~\cite{BeckFiala-DAM81} ($s$-sparse binary vectors),
and general vectors in the unit ball~\cite{Banaszczyk-Journal98,Beck-Combinatorica81,Giannopoulos,Rothvoss14,BansalDGL18}.

There has been a renewed interest in the online discrepancy setting, where many of
the techniques developed for the offline setting no longer
extend. Most of the results for online vector discrepancy deal with
\emph{stochastic settings} of the problem where the arriving vectors
satisfy some distributional
assumptions~\cite{Ajtai98,BS20,BJMSS-SODA21}. A
recent breakthrough work~\cite{ALS-STOC21} gives a very elegant
randomized algorithm with near-optimal discrepancy for general vectors
arriving online. However, to the best of our knowledge, none of these
ideas easily extend to the fully-dynamic setting where vectors can
also depart---which is the focus of this paper. In fact, we do not
know how to adapt existing ideas to establish non-trivial results for
even the simple \emph{deletions-only} setting: starting with $T$
vectors, a uniformly random subset of $\nicefrac{T}{2}$ of these vectors are
deleted one-by-one. Can we always maintain a low-discrepancy signing
of the remaining vectors with small recourse?

The study of dynamic algorithms also has a rich history, both in the
\emph{recourse} model, which measures the number of updates made the
algorithm per update, and the \emph{update-time} model, which measures
the running time of the algorithm per update. Apart from graph
problems, these models have been studied in a variety of settings such
network design~\cite{imase1991dynamic, gupta2014online, gu2016power,
  lkacki2015power}, clustering \cite{guo2020power, cohen2019fully},
matching \cite{grove1995online,chaudhuri2009online, bosek2014online},
and scheduling \cite{phillips1993online, westbrook2000load,
  andrews1999improved, sanders2009online, skutella2010robust,
  epstein2014robust, gupta2014maintaining}, and set cover
\cite{bhattacharya2018deterministic, bhattacharya2017deterministic,
  bhattacharya2019deterministically, abboud2019dynamic,
  bhattacharya2019new, bhattacharya2020improved,
  Gupta:2017:ODA:3055399.3055493}.

A different version of
edge-orientation, commonly known as graph balancing, involves
minimizing just the maximum in-degree (see,
e.g.,~\cite{MR1734116,MR2313458,MR3238400}): the techniques used for
that version seem quite different from those needed here.

\paragraph{Paper Outline.} We present the results for
\fdvb and specifically \Cref{thm:main1} in \S\ref{sec:distrib-BG}. The
results for graph balancing appear in \S\ref{sec:expand-decomp}. Other results for local-search algorithms appear in
\S\ref{section-lower-bounds-local-search} and
\S\ref{sec:local-search-upper}. We close with an insertion-only algorithm
for sparse vectors, and  conclusions and open problems  in \S\ref{sec:insert-only}.

\section{Fully-Dynamic Vector Balancing}
\label{sec:distrib-BG}

In this section, we prove~\Cref{thm:main1}. Given a set of vectors
$a_1, a_2, \ldots, a_T \in [-1,1]^n$, the \bargrin algorithm signs
them such that the discrepancy of the signed sum is at most
$2n$. However, this signing is highly sensitive to insert or delete
operations.  We address this issue by recursively dividing the input
sequence such that we lose only $O(n)$ discrepancy at each level of
this recursion tree---we call this the distributed \bargrin
algorithm. We then show how it can easily handle insert and delete
operations with low recourse.

The main idea underlying the 
\bargrin algorithm~\cite{BaranyGrinberg81} is the following linear algebraic lemma.
\begin{lemma}[Rounding Lemma~\cite{BaranyGrinberg81}]
  \label{lem:BG}
  Let $a_1, a_2, \ldots, a_T \in [-1,1]^n$ be the columns of matrix
  $A \in [-1,1]^{n \times T}$. For any initial fractional signing
  $x \in [-1,1]^T$, there exists a (near-integral) signing $y$ with
  all but $n$ variables being $\pm 1$ such that $A y = A
  x$. 
\end{lemma}
The signing $y$ is obtained by moving to a basic feasible solution
(BFS) of the following set of linear constraints
$ \{Ay = Ax, \; y \in [-1,1]^T\}$, where $x$ is treated as being
fixed.  Based on \Cref{lem:BG}, B\'ar\'any and
Grinberg~\cite{BaranyGrinberg81} gave the following offline algorithm:
starting with the all-zeros vector as the fractional signing (i.e.,
$x = \bf{0}$), let $y$ be the almost-integral vector satisfying
$A y = 0$. Now 
randomly rounding the fractional
variables (with bias given by the $y_i$ values) and using
concentration bounds shows a discrepancy of $O(\sqrt{n \log n})$, or
using sophisticated rounding schemes can give the tight $O(\sqrt{n})$
discrepancy~\cite{Spencer85,Bansal-FOCS10,Lovett-Meka-SICOMP15}. 

\subsection{An Equivalent, Recursive Viewpoint}
\label{sec:dist-BG}

A natural question is: \emph{can we extend the above \bargrin algorithm to
  the dynamic case}? Naively using the rounding lemma does not work,
since the rounded solutions $y$ and $y'$ for matrices $A$ and $A'$
differing in one column could be very different. Our idea is to
simulate the \bargrin algorithm in a distributed and recursive manner. We divide the sequence $\{1, \ldots, T\}$ into
sub-sequences of length $2n$ each, which gives us a set of $m:= T/2n$
sub-sequences (assume w.l.o.g., e.g., by padding, that $T/2n$ is a
power of 2). Let $P_1, \ldots, P_m$ denote these sub-sequences ordered
from left to right. We build a binary tree $\calT$ of height
$\log_2 m$ on $m$ leaves, where leaf $j$ corresponds to the
sub-sequence $P_j$. Similarly, for an internal node $v$, define $P_v$ 
as the sub-sequence formed by taking the union of $P_j$ over all leaves
$j$ below $v$.

The signing algorithm $\DBG(v)$, where $v$ is a node of $\calT$ is
shown in~\Cref{algorithm:recursivebg}. It assigns values
$y_i^v \in [-1,1]$ to the vectors $a_i$ for $i \in P_v$ such that the
following two conditions are satisfied: 
\begin{enumerate}  [partopsep=0pt,topsep=0pt,parsep=0pt]
	\item[(I1)] $\sum_{i \in P_v} y_i^v a_i = 0$, and 
	\item[(I2)] all but at most $n$
variables $y_i^v, i \in P_v$ are either $+1$ or $-1$.  
\end{enumerate}
Applying this
property to the root node yields~\Cref{lem:BG}. While the end result
is identical to the one-shot \bargrin algorithm, this yields some
crucial advantages in the dynamic setting.  Indeed, when a vector is
inserted/deleted, only a single leaf's sub-sequence changes.  We will
show that this leads to making changes in the signing assigned by the
ancestors of just this leaf, giving a total recourse of
$\widetilde{O}(n)$ per update!

For a subset $I$ of indices, let $A_I$ denote the submatrix of $A$ given by the columns corresponding to $I$. Similarly, for a vector $z$ indexed by $P_v$ and a subset $F$ of $P_v$, define $z|_F$ as the restriction of $z$ to $F$. The algorithm $\DBG(v)$ begins by recursively assigning values to the two sub-sequences corresponding to its two children. Since these assignments, satisfy the two invariant conditions above, combining the two solutions into a new solution $x$ (in line~\ref{l:BGcombine}) leads to at most $2n$ fractional variables. Using~\Cref{lem:BG}, we reduce the number of fractional variables to $n$. 
Finally, we can maintain a (integral) signing $\sigma$ by randomly assigning signs to  the fractional variables $F_r$ at the root $r$ and retaining the values $y^r$ for rest of the vectors. We now show by induction that the two invariant properties are satisfied, the proof is deferred to~\Cref{sec:appsec2}. 
\begin{restatable}{lemma}{dbginv}
  \label{lem:dbginvariant}
  The variables $y_i^v, i \in P_v$ satisfy the invariant properties~(I1) and (I2) at the end of $\DBG(v)$.  
\end{restatable}


\begin{algorithm}
  \caption{Distributed-\bargrin: $\DBG(v)$}
  \label{algorithm:recursivebg}
  \textbf{Input:} A node $v$ of $\calT$. \\
  \textbf{Output:} $(y^v, F_v)$: an assignment $y^v_i \in [-1,1]$ for each $i \in P_v$, and $F_v \subseteq P_v$ is the index set of ``fractionally'' signed vectors, i.e., indices $i$ such that $-1 < y^v_i < 1$. 
  \begin{algorithmic}[1]
    \If{$v$ is not a leaf}
    \State Let $v_L$ and $v_R$ be the left and the right children of $v$ respectively. 
    \State $(y^{v_L}, F_{v_L}) \leftarrow \DBG(v_L), (y^{v_R}, F_{v_R}) \leftarrow \DBG(v_R).$
    \State Define $F := F_{v_L} \cup F_{v_R}$, $x_i := y^{v_L}_i$ for all $i \in P_{v_L}, x_i := y^{v_R}_i$ for all $i \in P_{v_R}$.  \label{l:BGcombine}
    \Else 
       \State Define $F:= P_v, x_i = 0$ for all $i \in P_v$. 
    \EndIf
    \State Using~\Cref{lem:BG} find a vector $y' \in [-1,1]^{|F|}$
    such that (i) $A_{F} \cdot y'= A_F \cdot  x|_F$, and (ii) there
    are at most $n$ indices (denoted by the set $F_v \subseteq F$) having $-1 < y_i' < 1.$
    \label{l:bg}
    \State Define $y^v_i = x_i$ for $i \in P_v \setminus F$ and $y^v_i = y'_i$ for $i \in F$. 
    \State Return $(y^v, F_v)$. 
  \end{algorithmic}
  \end{algorithm}

\subsection{Dealing with Update Operations} 
Before describing the insert/delete operations, we describe a useful subroutine \updatevec, which given an assignment  $\yold$ to $a_1, \ldots, a_T$, updates it to a new assignment $\ynew$ when one of the vectors $a_j$ in the sequence changes.  The algorithm is very similar to~\Cref{algorithm:recursivebg}, but it needs to recurse on only one child of $v$. the one containing the index $j$. 
As a result, the vectors $\yold$ and $\ynew$ differ in at most $O(n\log T)$ coordinates. Details are deferred to~\Cref{sec:appsec2}.

\paragraph{Dynamic Insert and Delete.} We now discuss the algorithm when an insert or delete operation happens. The algorithm works in phases:
a new phase starts when the number of vectors becomes $2^\ell$ for some $\ell$, and ends when this quantity reaches $2^{\ell-1}$ or $2^{\ell+1}$. Whenever a new phase starts, we run $\DBG$ algorithm to find an assignment $y$. During a phase, we always maintain exactly $2^{\ell+1}$ vectors -- this can be ensured at the beginning of this phase by padding with $2^\ell$ zero vectors. This ensures that the tree $\calT$ does not change during a phase.

When a delete operation happens, we call \DBGUpdate, where the deleted vector gets updated to the zero vector. Similarly, when an insert operation happens, we update one of the zero vectors to the inserted vector. Thus, we get the following result:

\begin{lemma}
  The amortized recourse of this fully-dynamic algorithm is
  $O(n \log N)$ per update operation, where $N$ is the maximum number
  of active vectors at any point in time.
\end{lemma}

\begin{proof}
The work done at the beginning of a phase can be charged to the length of the input sequence at this time. This results in $O(1)$ amortized recourse. 
We show in~\Cref{cor:BG} that the amortized recourse after each update during a phase is $O(n \log N)$. This proves the overall amortized recourse bound. 
\end{proof}

Finally, since there are only $n$ fractional variables at the root,
we can use any state-of-the-art offline discrepancy minimization algorithm to
sign these vectors, e.g., to get $O(\sqrt{n})$ discrepancy for vectors
with unit $\ell_\infty$-norm \cite{Spencer85,Bansal-FOCS10,Lovett-Meka-SICOMP15}, or to get
$O(\sqrt{\log n})$ discrepancy for vectors with unit $\ell_2$-norm~\cite{BansalDG-FOCS16,BansalDGL18}.
This proves \Cref{thm:main1}.



\section{Fully-Dynamic Edge Orientation}
\label{sec:expand-decomp}

We next consider the case of dynamically orienting edges in a graph to
maintain bounded discrepancy.  In this problem, at each time/update an
adaptive adversary either inserts a new edge $e_t = (u,v)$ or removes
an existing edge $e$ from a graph $G(t)$. Assigning an \emph{orientation} to each edge
$(u,v)$ as $u \rightarrow v$ or $v \rightarrow u$, the discrepancy of
a vertex $v$ is
$\disc(v) = |~ |\delta_{\rm in}(v)| - |\delta_{\rm out}(v)| ~ |$,
where $\delta_{\rm in}(v)$ and $\delta_{\rm out}(v)$ are the sets of
in- and out-edges incident at $v$. Our goal is to minimize $\max_{v \in V} \disc(v)$.  The algorithm is allowed to
re-orient any edge $e$, and the \emph{amortized recourse} is the
average number of re-orientations per edge insertion/deletion. We now
present the first fully-dynamic algorithms with ${\rm polylog}(n)$
discrepancy and recourse.

\paragraph{Useful Notation}
For an undirected graph $G$ and any set $S \sse V$, define $E(S)$ as
the set of edges whose endpoints are both in $S$; for sets $S, T$,
define $E(S,T) = \{ e \in E \mid |e \cap S| = |e \cap T| = 1\}$.
Define the \emph{volume} of a set $S$ to be
$\vol(S) = \sum_{v \in S} \deg(v)$.

\begin{definition}[$\phi$-Expander]
	\label{def:expander}
	A graph $G$ is a \emph{$\phi$-expander} if for all subsets $S\subseteq V$, 
	\[
	|E(S,V-S)|\geq \phi \cdot \min(\vol(S),\vol(V-S)) \enspace . 
	\]
	In this case, we also say the graph $G$ has \emph{conductance} at least $\phi$.
\end{definition}

\begin{definition}[$\gamma$-Weak-Regularity]
	\label{def:weak-reg}
	For $\gamma \in [0,1]$, an undirected graph $G$ is
	\emph{$\gamma$-weakly-regular} if the minimum degree of any vertex is at
	least $\gamma$ times the average degree $2m/n$.
\end{definition} 


\subsection{High Level Overview} \label{sec:graph-overview}
We now provide a detailed overview of our algorithm, and then delve into the individual components.
A natural algorithm for the edge orientation problem is a {\em local search}
procedure: while there exists an edge $(u,v)$ currently oriented $u \rightarrow v$ such that
 ${\rm disc}(v) > {\rm disc}(u) + 2$, flip its orientation to $v \rightarrow u$. 
Although  locally optimal orientations 
could have discrepancy $\Omega(n^{1/3})$ for general graphs (see an example in~\Cref{sec:local-search-graph-lb}), our first crucial result is that 
they always have low discrepancy on \emph{expanders}.

\begin{theorem}
	\label{theorem:local-search-expander}
	Let $G(V,E)$ be a $\gamma$-weakly-regular $\phi$-expander. Then the
	discrepancy of any solution produced by \local is
	$O\big(\frac{\log m}{\phi \gamma}\big)$.
\end{theorem}

The proof of this theorem appears in~\Cref{sec:localExpander}. In
order to apply our local search algorithm to arbitrary graphs, our
plan is to use the powerful idea of expander decompositions (see,
e.g.,~\cite{ST-STOC04,saranurak2019expander,bernstein2020fullydynamic}).
At a high level, such schemes decompose any graph $G$ into a disjoint union of expanders with each vertex appearing in a small number of them. For concreteness, we use the following result
from~\cite[Theorem~19]{GuptaKKS20}\footnote{For ease of
	exposition, we use a result that runs in exponential time;
	using approximate low-conductance cuts gives a polynomial runtime
	with additional
	logarithmic factors.}.

\begin{theorem}[Decomposition into Weakly-Regular Expanders]
	\label{theorem:gkks20}
	Any graph $G = (V,E)$ can be decomposed into an edge-disjoint union
	of smaller graphs $G_1, G_2, \ldots,G_k$ such that: (a) each vertex
	appears in at most $O(\log^2 n)$ many smaller graphs, and (b) each
	of the smaller subgraphs $G_i$ is a $\phi/4$-weakly-regular
	$\phi$-expander, where $\phi=\Theta(1/\log n)$.
\end{theorem}

In order to make this into a dynamic decomposition, our algorithm follows a  natural idea of maintaining $\log m$ \emph{levels/scales}, and placing each edge of the current graph $G$ at one of these levels. We use $G_i$ to denote the
subgraph formed by the level-$i$ edges; crucially, we ensure that $G_i$ has at
most $2^i$ edges. For each level $i$, we maintain the expander decomposition of $G_i$ into $\bigcup_{j} \, G_{i,j}$ where $G_{i,j}$ represents the $j^{th}$
expander in this decomposition. Since each vertex appears in at most $\log^2 n$ expanders at every level, overall any vertex will appear in $O(\log^3 n)$
expanders. Hence, our goal is to maintain a low-discrepancy signing for each expander, with bounded number of re-orientations as the expander changes due to updates. Next we discuss how insertions are easier to handle, but deletions require several new ideas.

{\bf Insertions.} When edges are inserted into $G$, we insert it into $G_1$ (the lowest scale) and orient it arbitrarily. Whenever a level $i$ becomes full, i.e., $|G_i| > 2^i$, we remove all edges and add them to the higher level $i+1$, and \emph{recompute the expander decomposition using~\Cref{theorem:gkks20} from scratch for the graph consisting of all edges in this level}. We also recompute an optimal offline low-discrepancy discrepancy orientation for each expander\footnote{It is easy to optimally orient any graph in the offline setting: we consistently orient the edges of all cycles, to be left with a forest. We can then again orient all the maximal paths between pairs of leaves in a consistent manner, to end up with an orientation where every vertex has discrepancy in $\{-1,0,1\}$.}. Of course, we may need to cascade to higher levels if the next level also overflows. However, the total cost of all these edge reorientations can be easily charged to the recent arrivals that caused the overflow.

{\bf Deletions.} Our insertion procedure guarantees that an expander $G_{i,j}$ only observes deletions in its lifetime (before the expander decomposition at its level is recomputed). So when the adversary deletes an edge from $G$ (called a {\em primary deletion}), we can remove it from the corresponding expander $G_{i,j}$ it belongs to, and simply re-run local search from the current orientation if it continues to have expansion at least, say $\phi/6$. We can then bound the recourse by tracking the changes to the associated $\ell_2$ potential $\Phi$ for this graph. However, \emph{what do we do when $G_{i,j}$ ceases to be an expander?} 
Our idea is to simply identify a cut  of sparsity $< \phi/6$ and remove the smaller side $\Delta P$ from the graph $G_{i,j}$, and repeat if necessary. 
This is called the $\Prune$ procedure and we formally describe it in~\Cref{subsection:dynamic-expander-pruning}. The edges which are incident to $\Delta P$ are re-inserted into the system using the insertion algorithm. In~\Cref{theorem:pruning-ravi}, we bound the number of pruned edges (also called \emph{secondary deletions}) in terms of the number of actual adversarial edge deletions which caused the drop in expansion, and so we are able to amortize the recourse of re-inserting these pruned edges back into our algorithm.

\begin{theorem}
	\label{theorem:pruning-ravi}
	Let $G_0 = (V_0, E_0)$ be a $\phi$-expander with $m$ edges, $n$
	vertices, and minimum degree $\delta$. 
	For a subset $S \sse V_0$, let $\vol_0(S)$ denote its initial volume
	in $G_0$.   There is an algorithm called \Prune (described in~\Cref{subsection:dynamic-expander-pruning}), which for every adversarial deletion of any edge in $G_0$, outputs a (possibly empty) set of vertices $\Delta P$ to be pruned/removed which satisfies the following properties. 
	
	Let $P_t$ denote the aggregate set of vertices pruned over a sequence of $t$ adversarial deletions inside $G_0$, i.e., $V_t := V_0 \setminus P_t$ and $G_t$ is the graph with the undeleted edges of $E_0$ that are induced on $V_t$. Then, for each $1 \leq t \leq \phi^2 m/20$:
	\begin{enumerate} [label=(\roman*),partopsep=0pt,topsep=0pt,parsep=0pt,itemsep=0pt] 
		\item $P_t \subseteq P_{t+1}$. \label{cond:PIncreases}
		\item $G_t$ is a $\phi/6$-``strong expander'', i.e., for any subset
		$A \subseteq V_t$,
		$$ |E_t(A, V_t \setminus A)| \geq (\phi/6)\cdot
		\min\big(\vol_0(A), \vol_0(V_t \setminus A) \big). $$ Hence
		the minimum degree of a vertex in $V_t$ is at least
		$\phi\delta/6$.
		\item $\vol_0(P_t) \leq 6t/(5\phi)$. \label{cond:volP}
	\end{enumerate}
\end{theorem}
Similar ideas have recently been used for dynamic graph algorithms, e.g. in~\cite{saranurak2019expander,bernstein2020fullydynamic}, but our algorithms and analyses are more direct since we are concerned only with the amortized recourse rather than the update time. However, new  challenges appear due to our discrepancy minimization setting.

{\medskip \noindent \bf A `Potential' Problem. } 
While the above procedure identifies the set of edges to prune, so that the residual graph remains an expander, we still need to maintain a low-discrepancy orientation on the expander as it undergoes deletions and prunings. Indeed, the above ideas essentially allow us to cleanly reduce the fully dynamic problem to the follow special case of only handling deletions on expanders: let $G = (V,E)$ be a $\phi/6$-expander, currently oriented according to local search. Then, suppose $e$ is an adversarial deletion, and suppose $\Delta P$ is the set of vertices to be removed as computed by the $\Prune$ procedure. Then, how many flips would we need to end-up at a locally-optimal orientation on $G[V \setminus \Delta P]$, which we know has bounded discrepancy since $G[V \setminus \Delta P]$ is an expander? If we can bound this in terms of the number of edges incident to $\Delta P$, then we would be done, since these are precisely the number of secondary deletions, which are in turn bounded in terms of adversarial deletions. 

A natural attempt is to simply re-run local search on $G[V \setminus \Delta P]$ starting from the current orientation. While this will converge to a low-discrepancy solution because $G[V \setminus \Delta P]$ is an expander, our recourse analysis proceeds by tracking a  quadratic potential function, and this could increase a lot if we suddenly remove all edges incident to $\Delta P$ en masse. Removing the edges one by one is also also an issue as the intermediate graphs won't satisfy the desired expansion to argue both discrepancy as well as recourse (which indirectly depends on having good discrepancy bounds to control the potential). To resolve this issue, we craft a collection of ``fake'' intermediate graphs that interpolate between the graphs $G$ and $G[V \setminus \Delta P]$ which ensure that (i) all of them 
have good expansion properties, and (ii) the potential change in moving from one to another is bounded. Our overall algorithm is to then repeatedly re-run local search after moving to each intermediate graph, until we end up with the final orientation on $G[V \setminus \Delta P]$. 

We now formalize this in the following theorem, which bounds the recourse needed to move from a locally optimal orientation in $G[V]$ to one in $G[V \setminus \Delta P]$. Let $H$ denote any graph with a current orientation represented by $\dir{H}$. We then define the following potential 
\[ \Phi(\dir{H}) := \sum_{v \in V(H)} \disc(v)^2 \, .\]

\begin{theorem} \label{thm:iterated-local-search}
	Let $G_{t-1} = (V_{t-1}, E_{t-1})$ be a $\phi/6$-expander as maintained by our algorithm, and suppose the adversary deletes an edge $e_t \in E_{t-1}$. Moreover, suppose an associated set of vertices $\Delta P \subseteq V_{t-1}$ are pruned by $\Prune$ to obtain the graph $G_{t} = (V_t, E_t)$ which is a $\phi/6$-expander, where $V_t = V_{t-1} \setminus \Delta P$ and $E_{t}$ is the subset of $E_{t-1} \setminus \{e_t\}$ induced on $V_t$. 
	
	Then, starting from a locally optimal orientation $\dir{G}_{t-1}$ we can compute a locally optimal orientation $\dir{G}_t$ by performing at most $L_t$ flips satisfying
	$$ L_t \leq \left(\frac{ \log m}{\phi^2 \gamma}  + \frac{\vol_0(\Delta P) \log m}{\phi^2 \gamma}\right) + \Phi(\dir{G}_{t-1}) - \Phi(\dir{G}_{t})  \, . $$
\end{theorem}

With this our algorithm description is complete. For the discrepancy analysis, note that our algorithm at all times maintains a locally optimal orientation in each expander at each level, and every vertex appears in at most ${\rm polylog}(n)$ expanders from~\Cref{theorem:gkks20}, giving us an overall discrepancy of ${\rm polylog}(n)$ by combining with~\Cref{theorem:local-search-expander}. For the recourse analysis, any time the insertion algorithm overflows and a rebuild happens in the higher level, we can charge the recourse to the adversarial insertions as well as re-insertions of the edges removed by $\Prune$. The latter is in turn bounded in terms of the adversarial deletions by~\Cref{theorem:pruning-ravi}. Finally, we bound the total recourse within an expander, as parts of it are pruned out, for which we appeal to~\Cref{thm:iterated-local-search}. Since~\Cref{theorem:pruning-ravi}\,(iii) ensures that the total volume of all the sets which are pruned can be bounded in terms of $O(1/\phi)$ times the number of adversarial deletions, we get that the total number of flips done over a sequence of adversarial deletions in any expander is at most $\frac{\log m}{\phi^3 \gamma}$ times the number of adversarial deletions plus the potential $\Phi(\dir{G}_0)$ of the initial expander, which is small since we start with an optimal orientation where each vertex has discrepancy at most $1$ when the expander is formed.

\subsection{Local-Search for Weakly-Regular Expanders} \label{sec:localExpander}

In this section we prove \Cref{theorem:local-search-expander} that local search ensures low discrepancy on any weakly-regular expander. 
Recall that the local search flips an edge $(u,v)$ oriented from $u$ to $v$ whenever $\disc(v) > \disc(u) + 2$.   

\begin{algorithm}
	\caption{\local}
	\label{algorithm:local}
	\textbf{Input:} Graph $G=(V,E)$ and an initial partial orientation.\\
	\textbf{Output:} Revised orientation which is a local optimum.
	\begin{algorithmic}[1]
		\State Arbitrarily direct any undirected edges in $G$.
		\State While there exists a directed edge $(u,v)$ such
		that flipping it decreases $\Phi:= \sum_u \disc(u)^2$, flip
		it. 
	\end{algorithmic}
\end{algorithm}

\begin{proof}[Proof of~\Cref{theorem:local-search-expander}]
	Let $\vec{G}=(V, \vec{E})$ be the directed graph
	corresponding to a local optimum.  Consider the node $v$ with
	largest discrepancy $k$; without loss of generality, assume
	$k \geq 0$. We perform a breadth-first-search (BFS) in $\vec{G}$
	starting from $v$, but only following the \emph{incoming} edges at each step. Let
	$L_i$ be the vertices at level $i$ during this BFS, i.e., $L_i$ is
	the set of vertices $w$ for which the shortest path in $\vec{G}$ to
	$v$ contains $i$ edges. Let $S_i$ denote the set of vertices up to
	level $i$, i.e., $S_i := \bigcup_{i'=0}^i L_{i'}$. The fact
	that $\vec{G}$ is a local optimum means there are no improving flips,
	and hence the discrepancy of any vertex in $L_i$ is at least
	$k-2i$. In turn, this implies that there are at least $k/2$ layers,
	and the discrepancy of any vertex in $S_{k/4}$ is at least $k/2$.
	We now show that the volume of $S_{k/4}$'s complement is large.
	
	\begin{claim} \label{claim:largeVol}
		$\vol(V \setminus S_{k/4}) \geq 2\gamma m/3$. 
	\end{claim} 
	
	\begin{proof}
		Each node in $S_{k/4}$ has discrepancy at least $k/2$, and each
		node in $V\setminus S_{k/4}$ has discrepancy at least $-k$. Since
		the total discrepancy of all the vertices in $V$ is 0, it follows
		that
		\[ 
		0 ~=~ \sum_{v \in V} \disc(v) ~\geq~ |S_{k/4}| \cdot k/2 - |V \setminus S_{k/4}| \cdot k. 
		\]
		This implies that $|V \setminus S_{k/4}| \geq n/3.$ Now
		$\gamma$-weak-regularity implies each vertex in $G$ has degree at
		least $\gamma \cdot \frac{2m}{n}$, and hence the sum of the degrees
		of the vertices in $V \setminus S_{k/4}$ is at least $2m
		\gamma/3$. 
	\end{proof}

	We now show that the size of the edge set $E(S_i)$ increases geometrically.
	\begin{claim}
		\label{cl:expinc}
		For any $i \leq k/4$, $|E(S_{i+1})| \geq (1 + \phi \gamma/3) |E(S_i)|.$
	\end{claim}
	\begin{proof}
		Given a directed graph $\vec{G}$ and a subset $X$ of vertices, let
		$\delta^-(X)$ and $\delta^+(X)$ denote the set of incoming edges
		into $X$ (from $V \setminus X$), and the set of outgoing edges from
		$X$ (to $V \setminus X$) respectively. Since the discrepancy of each
		vertex in $S_i$ is positive,
		\[ 
		0 ~\leq~ \sum_{w \in S_i} \disc(w) ~=~ \sum_{w \in S_i}
		\big( |\delta^-(w)| - |\delta^+(w)| \big) 
		~=~|\delta^-(S_i)| - |\delta^+(S_i)|. 
		\]
		The expansion property now implies that 
		\begin{align}
			\label{eq:expansion1}
			|\delta^-(S_i)| ~\geq~ \tfrac12 |\delta(S_i)| ~\geq~ \tfrac{\phi}{2}\cdot    \min\big(\vol(S_i), \vol(V-S_i) \big).
		\end{align}
		Since all edges in $\delta^-(S_i)$ are directed from $L_{i+1}$ to $L_i$,
		\begin{align*}
			|E(S_{i+1})| ~ \geq~ |E(S_i)| + |\delta^-(S_i)|
			\; &~\stackrel{\eqref{eq:expansion1}}{\geq} ~
			|E(S_i)| + \frac{\phi}{2}   \min\big(\vol(S_i), \vol(V-S_i) \big) \\ 
			&~\geq~ |E(S_i)| + \frac{\phi}{2}   \min\big(|E(S_i)|, 2\gamma m/3\big), 
		\end{align*}
		where we used $\vol(S_i) \geq |E(S_i)|$ and \Cref{claim:largeVol}
		for the two terms of the last inequality. Since $m \geq |E(S_i)|$ and
		$\gamma \leq 1$, the RHS above is at least $ |E(S_i)| \cdot (1 + \phi \gamma/3). $
	\end{proof}
	
	Finally the fact that $S_0 = \{v\}$ implies that 
	$|E(S_1)| \geq  |\delta^-(S_0)| \geq k$, and using \Cref{cl:expinc} we get
	$$|E(S_{k/4})| \geq k \; \big( 1 + \nicefrac{\phi \gamma}{3} \big)^{k/4-1}. $$ Since $|E(S_{k/4})| \leq m$, we get 
	$k = O\big(\frac{\log m}{\phi\gamma}\big)$.
\end{proof}

\begin{figure}
	\centering
	\includegraphics{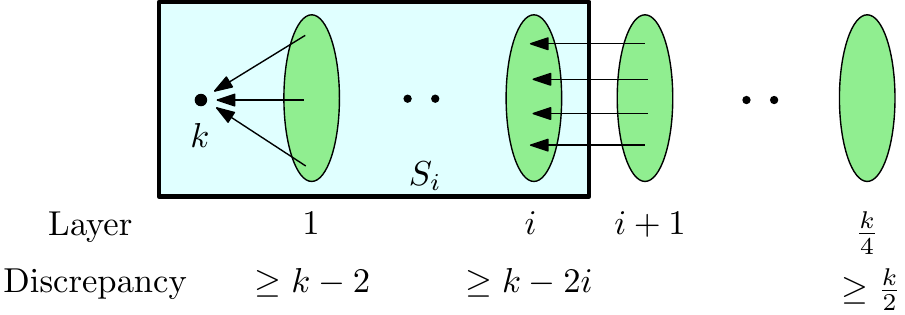}
	\caption{local-search for expanders.}
	\label{fig:expanders-local-search}
\end{figure}

The weak-regularity property was used only in \Cref{claim:largeVol}
above; it is easy to alter the proof to show that
$\vol(V \setminus S_{k/4}) \geq \nicefrac{\gamma m}{4}$ even if all but $\nicefrac{n}{7}$
vertices in $G$ satisfy the weak-regularity property. This proves~\Cref{theorem:local-search-expander}.

\begin{corollary}
	\label{cor:weakerlocalsearch}
	Let $G$ be a $\phi$-expander, such that degree of every vertex,
	except perhaps a subset of at most $\nicefrac{n}{7}$ vertices, is at least
	$ 2\gamma m/n$. Then the discrepancy of \local is
	$O \big( \frac{\log m}{\phi \gamma} \big)$.
\end{corollary}

The expansion plays a crucial role here:
in \Cref{section-lower-bounds-local-search} we show that locally-optimum
solutions of \local can have large discrepancy for general
graphs.

\subsection{Dynamic Expander Pruning}
\label{subsection:dynamic-expander-pruning}

In this section we prove~\Cref{theorem:pruning-ravi}. We first recall the setup. Suppose we start with a $\phi$-expander $G_0=(V_0, E_0)$, and at each
step $t$ some edge $e_t$ is deleted by the adversary (call these \emph{primary
	deletions}). The goal is to remove a small portion of the graph so that the remaining portion continues to be, say, a $\phi/6$-expander.  
Since the graph may may violate the expansion requirement due to deletions, we perform additional \emph{secondary
deletions} at each step to maintain a slightly smaller subgraph
$G_t = (V_t, E_t)$ which is $\nf{\phi}{6}$-expanding, such that
$\{e_1, \ldots, e_t\} \cap E_t = \emptyset$.  Crucially, the number of
secondary deletions is only a factor $\nicefrac1\phi$ more than the
number of primary deletions until that point. The idea of this greedy {\em expander pruning algorithm} is simple:
whenever the edge deletion $e_t$ creates a sparse cut in $G_{t-1}$, we
iteratively remove the smaller side of such a sparse cut from the
current graph until we regain expansion. (See~\Cref{algo:prune} for
the formal definition. Again we assume we can find low-conductance
cuts; using an approximation algorithm would lose logarithmic factors.)

\begin{algorithm}
	\caption{$\Prune(G_{t-1} , e_t)$}
	\label{algo:prune}
	\textbf{Input:} Graph $G_{t-1} = (V_{t-1},E_{t-1})$ and an edge $e_t \in E_{t-1}$  which gets deleted at this step.\\
	\textbf{Output:} A set of vertices $\Delta P$ that have to be pruned from $G_{t-1}$ to get $G_t$.
	\begin{algorithmic}[1]
		\State Define $G_t$ with edges $E_t \gets E_{t-1} \setminus e_t$ and vertices $V_t \gets V_{t-1}$.
		\State $\Delta P \gets \emptyset$.
		\While{there is a subset $A \subseteq V_t$ with $|E_t(A, V_t \setminus A)| < (\phi/6) \cdot \vol_0(A)$}
		\label{l:sparse}
		\State Assume that $A$ is the smaller side of the cut
		\State Remove $A$ from $G_t$, i.e., $V_{t} \leftarrow V_t \setminus A$, $E_{t} \leftarrow G_t[V_t \setminus A]$, and $\Delta P \leftarrow \Delta P \cup A$.
		\EndWhile 
		\State Return  $\Delta P$
	\end{algorithmic}
\end{algorithm}

For a subset $S \sse V_0$, let $\vol_0(S)$ denote its initial volume
in $G_0$. Observe that the expansion in line~\ref{l:sparse} above is
measured with respect to $\vol_0(A)$, the volume of the set $A$ in
$G_0$, which is a stronger condition than comparing to the volume in
the current graph.
Define $P_t$ to be the set of vertices pruned in the first $t$
iterations, i.e., $P_t = V_0 \setminus V_{t}$, so that
$P_0 = \emptyset$. We now show that this algorithm maintains a
``strong expansion'' property at all times (i.e., the expansion
property holds with respect to the initial volume $\vol_0$).

\begin{proof}[Proof of~\Cref{theorem:pruning-ravi}]
	The first property uses that $P_t$ is the set of vertices removed in
	the first $t$ steps. The second property uses the stopping condition
	of the algorithm, and the fact that $\vol_0(\{v\}) \geq \delta$ for
	all vertices.
	
	To prove the final property, let the subsets pruned by the algorithm
	be $A^1, A^2, \ldots, A^s$ in the order they are pruned. For an
	index $\ell$, let $B^\ell := \cup_{\ell'=1}^\ell A^{\ell'}$. Let
	$E_t' := E_0 \setminus \{e_1, \ldots, e_{t} \}$. The following claim
	bounds the number of edges leaving $B^\ell$.
	
	\begin{claim} \label{claim:CutPEdges} If set $A^\ell$ is pruned in
		iteration $t$, then
		$|E_t'(B^\ell,V_0-B^\ell)| \leq \frac{\phi}{6}\vol_0(B^\ell)$.
		Also, $\vol_0(B^\ell) \leq \vol_0(V_0 \setminus B^\ell)$.
	\end{claim}
	
	Before we prove this claim, we use it to prove the third property: note that
	$ |E_0(B^\ell, V_0 \setminus B^\ell)| \geq \phi \cdot
	\vol_0(B^\ell) $ because $G_0$ is a $\phi$-expander, and moreover
	$\vol_0(B^\ell) \leq \vol_0(V_0 \setminus B^\ell)$ by the second
	part of \Cref{claim:CutPEdges}. Moreover, 
	the first part of the claim implies that 
	$ |E_t'(B^\ell, V_0 \setminus B^\ell)| \leq
	\frac{\phi}{6}\vol_0(B^\ell)$. Since $E_t'$ changes by at most
	one edge per deletion, it follows that
	\begin{gather}
		t ~\geq~ |E_0(B^\ell, V_0 \setminus B^\ell)| - |E_t'(B^\ell, V_0
		\setminus B^\ell)| \geq \phi \cdot \vol_0(B^\ell) -
		\frac{\phi}{6}\cdot \vol_0(B^\ell) ~=~ \frac{5 \phi}{6}
		\vol_0(B^\ell). \label{eq:1}
	\end{gather}
	Since $P_t$ is the same as $B^\ell$ for some $\ell$, this completes
	the proof.
\end{proof}

\begin{proof}[Proof of \Cref{claim:CutPEdges}]
	
		\begin{figure}
		\centering
		\includegraphics[width=4.5in]{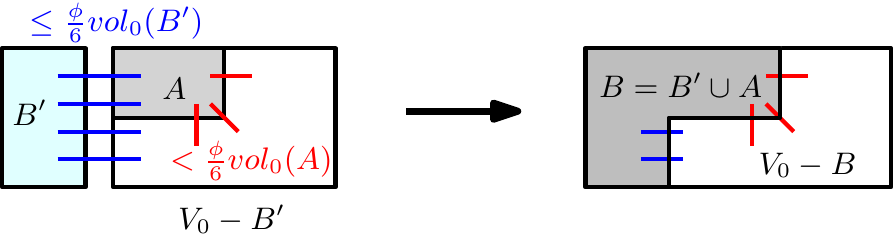}
		\caption{Dynamic expander pruning.}
		\label{fig:expander-pruning}
	\end{figure}
	
	We proceed by induction on $\ell$. When $\ell=0$ this is trivial
	since $B^\ell = \emptyset$. Suppose the claim holds for
	$B^{\ell-1}$, and we need to prove it for
	$B^\ell = B^{\ell-1} \cup A^\ell$. For sake of brevity, we denote
	$B' = B^{\ell - 1}, B = B^\ell$ and $A = A^\ell$ (see
	\Cref{fig:expander-pruning}). 	Any edge in $E_{t}'(B, V_{0} \setminus B)$ either lies in
	$E_t'(B', V_0 \setminus B) \sse E_t'(B', V_0 \setminus B')$ or in
	$E_{t}'(A, V_0 \setminus B)$. By the induction hypothesis, the
	former is at most $\frac{\phi}{6}\vol_0(B')$. The latter, by the
	condition in line~\ref{l:sparse}, is at most
	$(\phi/6) \cdot \vol_0(A)$. Summing the two, 
	we get $|E_{t}'(B, V_{0} \setminus B)| \leq (\phi/6)  \vol_0(B') +
	(\phi/6)  \vol_0(A)$, which proves the first part of the claim. 
	
	To prove second part of the claim, it suffices to show
	$\vol_0(B) \leq m$. By the induction hypothesis,
	$\vol_0(B') \leq m$. Further, by~(\ref{eq:1}) we have
	$\vol_0(B') \leq \frac{6t}{5 \phi}$. Now two cases arise:
	\begin{enumerate} [partopsep=0pt,topsep=0pt,parsep=0pt]
		\item $\vol_0(B')\geq \frac{\phi}{6}\vol_0(A)$: In this case, 
		\[ t~\stackrel{(\ref{eq:1})}{\geq}~ \frac{5\phi}{6}\vol_0(B')
		~\geq~ \frac{5\phi}{12}\vol_0(B') + \frac{5\phi^2}{72}\vol_0(A)
		~\geq~ \frac{5\phi^2}{72}\vol_0(B'+A) \; = \; \frac{5\phi^2}{72}
		\vol_0(B) \enspace.
		\]
		Since $t \leq \nicefrac{\phi^2 m}{20},$ it follows that
		$\vol_0(B) \leq m$.
		\item $\vol_0(B')< \frac{\phi}{6}\vol_0(A)$: Consider the cut
		$(A,V_0 \setminus B)$. Now,
		\begin{align*}
			|E_0(A,V_0 \setminus B)| &~=~ |E_0(A,V_0 \setminus
			A)|-|E_0(A,B')| ~\geq~ \phi
			\cdot \vol_0(A)- \vol_0(B') ~\geq~ \nicefrac56 \cdot \vol_0(A).
		\end{align*}
		On the other hand, $A$ is pruned by our algorithm because it is a
		sparse cut, i.e.,
		$|E'_t(A,V_0 \setminus B)| < \frac{\phi}{6} \cdot
		\vol_0(A)$. Therefore, the total number of deletions $t$ is at
		least
		$\frac{5\phi}{6}\vol_0(A)-\frac{\phi}{6}\vol_0(A)=\frac{4\phi}{6}\vol_0(A)$. Now
		we argue as in the first case. We get
		\[ t~>~ \frac{3\phi}{6}\vol_0(A) + \frac{\phi}{6}\vol_0(A) ~>~
		\frac{3\phi}{6}\vol_0(A) + \vol_0(B') ~ \geq ~ \frac{3\phi}{6}
		\vol_0(B'+A) ~=~ \frac{\phi}{2} \vol_0(B) \enspace.
		\]
		Since $t \leq \nicefrac{\phi^2 m}{20}$, we again get $\vol_0(B) \leq m.$    \qedhere
	\end{enumerate}
	
\end{proof}

\subsection{Dynamic Local-Search on Expanders with Deletions}
\label{subsection:local-search-pruning-expanders}

In this section we show how to  dynamically maintain a locally-optimal orientation of an expander, as parts of it are pruned out over time, thereby proving~\Cref{thm:iterated-local-search}. The algorithm appears
as~\Cref{algorithm:prune-fake-recolor}. 
We assume that the expander $G_{t-1}$ is maintained by dynamic pruning procedure $\Prune$ and satisfies the expansion properties of~\Cref{theorem:pruning-ravi}. We also assume that we have a locally optimal orientation $\dir{G}_{t-1}$ inductively maintained by~\Cref{algorithm:prune-fake-recolor}. Then, when the adversary deletes an edge $e_t$ and $\Prune$ computes a set $\Delta P$ of vertices to remove from $G_{t-1}$ to obtain a graph $G_t$, we show how to compute a locally optimal orientation $\dir{G}_t$ with a bounded number of flips.

Recall that we do this via a potential function argument. For any graph $H$ and current orientation $\dir{H}$, the potential of this orientation is $\Phi(\dir{H}) := \sum_{v \in V(H)} \disc(v)^2$. Indeed, the main issue is that there could be some vertices in $V_{t-1} \setminus \Delta P$ which are incident to many edges from $\Delta P$. Hence, if we remove $\Delta P$ in one shot, the potential of the residual graph could increase a lot. To resolve this, we replace $\Delta P$ by a set $F$ of an equal number $|\Delta P|$ of \emph{fake vertices}, and replace all the edges between $V_t := V_{t-1} \setminus \Delta P$ and $\Delta P$ with edges between $V_t$ and $F$ in a balanced round-robin manner to preserve the discrepancy of every vertex of $V_t$ w.r.t. its discrepancy in $\dir{G}_{t-1}$. Due to this balanced way of distributing the edges, we can show that the potential of the fake graph over $V_t \cup F$ is no more than that of $\dir{G}_{t-1}$, and moreover, even after deleting a subset $F' \sse F$ of fake vertices, $G_t \cup (F \setminus F')$ is an expander. These properties motivate running the following algorithm: transition from $G_{t} \cup F$ to $G_t$ by removing the fake vertices (and its incident edges) one-by-one, and re-running local search \emph{after each deletion}.

\begin{algorithm}[H]
	\caption{\textproc{Prune-and-Reorient}$(\dir{G}_{t-1}, e_t, \Delta P)$}
	\label{algorithm:prune-fake-recolor}
	\textbf{Input:}  Graph $G_{t-1}=(V_{t-1}, E_{t-1})$ with orientation $\dir{G}_{t-1}$, deleted edge
	$e_t$, and pruned set $\Delta P \sse V_{t-1}$.\\
	\textbf{Output:} A low-discrepancy orientation $\dir{G}_t$ for $G_t = (V_t, E_t)$ where $V_t := V_{t-1} \setminus \Delta P$ and $E_t$
	is the subset of $E_{t-1} \setminus \{e_t\} $ induced on $V_t$.  
	\begin{algorithmic}[1]
		\State Create fake vertices $F := \{f_1, \ldots , f_{N}\}$, where
		$N=|\Delta P|$, and define $H_t := (V_t \cup F, E_t)$. 
		\State Let $E^+$ be edges in $E_{t-1}(V_{t}, \Delta P) \setminus \{e_t\}$
		oriented from $V_{t}$ to $\Delta P$. Denote $E^+ = \{e_1, \ldots,
		e_r\}$, such that all edges incident to a vertex in $V_{t}$ appear consecutively.
		\For{$i=1, \ldots, r$}
		\State For edge $e_i = (v_i, p_i) \in E^+$, add edge $(v_i, f_{(i
			\bmod N)+1})$  into $H_t$ oriented from $v_i$ to $f_{(i \bmod N)+1}$. 
		\EndFor
		\State Repeat above loop for edges $E^-$ in $E_{t-1}(V_t, \Delta P) \setminus \{e_t\}$
		oriented into $V_t$; adds more edges to $H_t$.
		\State Run \local on $H_t$. \label{l:localsearchH} 
		\For{each $1 \leq j\leq N$}
		\State Remove vertex $f_j$ and incident edges from $H_t$. \label{l2:remfake}
		\State Run \local on (the current graph) $H_t$. \label{l:localsearchj}
		\EndFor
		\State Define the final orientation  $\dir{G}_t$ to be the final orientation $\dir{H}_{t}$ of $H_t$ (there are no fake vertices).
	\end{algorithmic}
\end{algorithm}


\begin{proof}[Proof of~\Cref{thm:iterated-local-search}] 
Let $H_t^j$ denote the graph $H_t$ after the
removal of the fake vertices $f_1, \ldots, f_{j}$, so that
$H_t^0 = H_t$. Since $t$ is fixed, we suppress the subscript $t$ for
the rest of this discussion.

We begin with some useful claims towards bounding the total number of flips $L_t$.
Firstly, we show that each of the intermediate graphs is a reasonable expander. The idea is that the pruned set $\Delta P$ is small compared to $V_{t}$, because whenever it becomes sufficiently large, the dynamic expander decomposition algorithm rebuilds the expander and we can charge the recourse to the adversarial deletions. As a result, both the number of fake vertices as well as their volume is substantially smaller than that of the ``real'' vertices $V_t$, and so the expansion properties of $V_t$ are approximately retained in each of the intermediate graphs $H^j$.

\begin{lemma}
	\label{lemma:expansion-fake-vertices}
	For each $j \in \{0, \ldots, N\}$, the graph $H^j$ is a $\phi/36$-expander.
\end{lemma}

\begin{proof}
	Let $F^j := \{f_{j+1}, \ldots, f_{N} \}$ be a suffix of the fake
	vertices. The vertex set of $H^j$ is $W^j := V_t \cup F^j$. Let
	$E^j$ denote the set of edges in $H^j$---these are the union of
	edges between $V_t$, i.e., $E_{t-1}(V_t)$, and those going between
	$F^j$ and $V_t$.  Let $S$ be a subset of vertices in $H^j$; we need
	to show that
	\begin{align}
		\label{eq:expproof}
		|E^j(S, W^j \setminus S)|~ \geq ~ \frac{\phi}{36} \min(\vol^j(S), \vol^j(W^j \setminus S)), 
	\end{align}
	where $\vol^j$ denotes the volume with respect to $E^j$. Let
	$\vol_t$ denote the volume with respect to the edges
	$E_t = E_{t-1}(V_t)$. Recall that $\vol_0$ denotes the volume with
	respect to $G_0$ (i.e., the expander graph before any deletions were
	performed).
	
	Without loss of generality, assume that
	$\vol^j(S) \leq \vol^j(W^j \setminus S)$. Let $S_r$ and $S_f$
	(``real'' and ``fake'') denote $S \cap A_t$ and $S \cap F^j$
	respectively (see \Cref{fig:fake-graphs-expansion}).
	
	\begin{figure}
		\centering
		\includegraphics{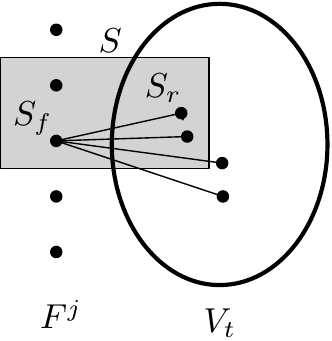}
		\caption{Expansion property of intermediate graphs $W^j$.}
		\label{fig:fake-graphs-expansion}
	\end{figure}
	
	\begin{claim}
		\label{cl:compl}
		$\vol_t(S_r) \leq \frac{2}{2-\phi} \cdot \vol_t(V_t \setminus S_r)$.
	\end{claim}
	\begin{proof}
		We know that $V_0 = P_t \cup V_t$. Since
		$t \leq \phi^2 m/20$, \Cref{theorem:pruning-ravi}(iii)
		shows that $\vol_0(P_t) \leq \frac{6}{5\phi} \cdot \frac{\phi^2 m}{20} \leq \phi m/4$. Therefore, $\vol^j(V_t)$ is at least $2\left(m - \frac{\phi m}{4}\right)=2m-\frac{\phi m}{2}$. 
		Now,
		$$\vol^j(S_r) ~\leq~ \vol^j(S) ~\leq~ \vol^j(W^j \setminus S) ~\leq~ \vol^j(V_t \setminus S_r) + \vol_0(P_t). $$
		Using $\vol_0(P_t) \leq \phi m/2$, this gives $\vol^j(S_r)-\vol^j(V_t\setminus S_r)\leq \phi m/2$. We also know that
		$\vol^j(S_r) + \vol^j(V_t \setminus S_r) = \vol^j(V_t) \geq 2m -
		\frac{\phi m}{2}$. 
		Eliminating $m$ from these two inequalities gives the claim. 
	\end{proof}
	\Cref{theorem:pruning-ravi} shows that 
	$$|E^j(S_r, V_t \setminus S_r)| ~\geq~ \frac{\phi}{6}
	\min(\vol_0(S_r),\vol_0(V_t \setminus S_r)) ~\geq~ \frac{\phi}{6} \cdot \frac{2-\phi}{2} \cdot \vol^j(S_r)~\geq~ \frac{\phi}{12} \vol^j(S_r),$$
	where the second-last inequality follows from \Cref{cl:compl}, and the last inequality uses $\phi\leq 1$.
	
	We now consider two cases:
	\begin{enumerate} [partopsep=0pt,topsep=0pt,parsep=0pt]
		\item $|E^j(S_f, V_t \setminus S_r)| \geq \vol^j(S_f)/2$: In this
		case,
		$$ |E^j(S, W^j \setminus S)|  ~\geq~  |E^j(S_r, V_t \setminus S_r)| +
		|E^j(S_f, V_t \setminus S_r)| ~\geq~ \frac{\phi}{12} \vol^j(S_r) +\frac{1}{2}\vol^j(S_f) ~\geq~ \frac{\phi}{12} \vol^j(S). $$ 
		
		\item $|E^j(S_f, V_t \setminus S_r)| \leq \vol^j(S_f)/2$: This
		implies $\vol^j(S_r) \geq \vol^j(S_f)/2$. Therefore,
		$$ |E^j(S, W^j \setminus S)| ~\geq~ |E^j(S_r, V_t \setminus S_r)|
		~\geq~ \frac{\phi}{12} \vol^j(S_r) ~=~ \frac{\phi}{36} (\vol^j(S_r) +
		2\vol^j(S_r)) ~\geq~ \frac{\phi}{36} \vol^j(S).$$
	\end{enumerate}
	Hence the proof of \Cref{lemma:expansion-fake-vertices} follows.
\end{proof}

Next we show that \local gives low discrepancy on the graphs $H^j$,
\emph{even though they may not be weakly-regular}.

\begin{lemma}
	\label{lemma:discrepancy-fake-vertices}
	For every $j \in \{0, \ldots, N\}$, the discrepancy of $H^j$ at a
	local optimum is $O \big(\frac{\log m}{\phi^2 \gamma} \big)$.
\end{lemma}
\begin{proof}
	We apply~\Cref{cor:weakerlocalsearch} to
	$H^j$. \Cref{lemma:expansion-fake-vertices} implies that $H^j$ is
	$\nf{\phi}{36}$-expander, so it suffices to show that a large fraction of
	the vertices of $H^j$ have large degree.
	
	First of all, since we allow at most 
	$D = \nf{\phi^2 m}{20}$ deletions, $\vol_0(P_t)\leq \frac{6}{5\phi}
	\cdot \frac{\phi^2 m}{20}=\frac{3\phi m}{50}$ by \Cref{theorem:pruning-ravi}(iii). Using $\gamma$-weak-regularity of $G_0$ and $\gamma=\phi/4$ (\Cref{theorem:gkks20}), this implies $|P_t|\leq \frac{n}{2m\gamma}\cdot \frac{3\phi m}{50} = \frac{3n}{25}$. So $H^j$ has at least $22n/25$ vertices. It follows that $|P_t|\leq 3n'/22 \leq n'/7$, where $n'$ denotes the number of vertices in $H^j$. Now, \Cref{theorem:pruning-ravi}(ii) implies that the degree of any vertex belonging to set $V_t$ in the graph $G_t$
	is at least $\phi \delta/6$, where $\delta \geq 2\gamma m/n$. (This
	uses that $G_0$ is $\gamma$-weakly-regular). Thus, the degree of any
	vertex in $V_t$ in the intermediate graph $H^j$ is also at least
	$$ \frac{\phi\delta}{6} ~=~ \frac{2 \gamma \phi  m}{6n} ~\geq~ \frac{22 \gamma \phi  m'}{75n'} ~=~ \left(\frac{11 \gamma \phi}{75}\right)\cdot \frac{2m'}{n'}, $$
	where $m',n'$ denote the number of edges and vertices in $H^j$
	respectively (since $m' \leq m$ and $n' \geq 22 n/25$ as shown
	above). The desired result now follows
	from~\Cref{cor:weakerlocalsearch}.
\end{proof}

We are now ready to conduct the potential-based analysis for bounding the number of flips. We bound the recourse by studying the $\ell_2$-potential $ \Phi(\dir{H}^j) := \sum_{w \in V(H^j)} \disc(w)^2 $ as we transition from $G_{t-1}$ to $G_t$. Indeed, note that a flip made by local search decreases the potential by at least 1, so the recourse is at
most the \emph{total increase} in the potential. This increase happens
during~\Cref{algorithm:prune-fake-recolor} when we replace $\Delta P$
by $F$ to get the graph $H_t$ (with its resulting orientation), and
when we remove the fake vertex $f_j$ from $H_t^{j-1}$ to get $H_t^j$
(in line~\ref{l2:remfake}). We bound the potential increase during each
of these steps. 

We give some notation first. Let
$\dir{G}_{t-1}, \dir{G}_t, \dir{H}_t^j$ be the subgraphs
$G_{t-1}, G_t, H_t^j$ oriented after running \local respectively. Let $\dir{G}'_{t-1}$ denote $\dir{G}_{t-1} \setminus \{e_t\}$. Recall that $E^+$, $E^-$ (and therefore $H_t$) are defined using edges of $\dir{G}'_{t-1}$. Let
$\dirp{H_t}$ be the orientation of $H_t$ just after we replace
$\Delta P$ by $F$ in $\dir{G}'_{t-1}$ (i.e., before
line~\ref{l:localsearchH}). Similarly let $\dirp{H^j}$ (we again
suppress the subscript $t$ for ease of notation) be the orientation of
$H^j$ just after we remove $f_j$ but before we run \local on it (in
line~\ref{l:localsearchj}). 
%
Since edges are added in a round-robin manner between $V_t$ and $\Delta P$ in
$H_t$, there are no parallel edges. 

\begin{claim}
	\label{cl:potential0}
	$\Phi(\dir{G}'_{t-1}) - \Phi(\dir{G}_{t-1})$ is at most $O\big( \frac{\log m}{\phi^2 \gamma}\big)$. 
\end{claim}
\begin{proof}
	Recall that $\dir{G}'_{t-1}$ is obtained by removing $e_t$ from $\dir{G}_{t-1}$. Let $d$ be the maximum discrepancy of a vertex in $\dir{G}_{t-1}$. \Cref{theorem:pruning-ravi} implies that $G_{t-1}$ is an $\Omega(\phi\gamma)$-weakly-regular $\Omega(\phi)$-expander, so \Cref{theorem:local-search-expander} implies that the discrepancy $d$ is $O \big( \frac{\log m}{\phi^2\gamma} \big).$ Hence the removal of $e_t$ from $\dir{G}_{t-1}$ can increase the potential by at most $2((d+1)^2 - d^2) = 2(2d+1) = O\big( \frac{\log m}{\phi^2 \gamma}\big)$, thus proving the claim. 
\end{proof}

Next, we show that the potential cannot increase while going from $\dir{G}'_{t-1}$ to $\dirp{H_t}$.  This uses the fact that we essentially re-distributed all the edges in $E^+$ and $E^-$ in a balanced round-robin manner. 
\begin{claim}
	\label{cl:potential1}
	$\Phi(\dirp{H_t}) - \Phi(\dir{G}'_{t-1}) \leq 0$.
\end{claim}
\begin{proof}
	For a given sum $s$ and variables satisfying $\sum_{i=1}^N x_i = s$, the optimal (w.r.t. $\ell_2$ norm) integer assignment of variables has $x_i \in \{\lfloor \frac{a_+ - a_-}{N} \rfloor, \lfloor \frac{a_+ - a_-}{N} \rfloor + 1\}$ for all $i$ and is unique up to permutations. For our problem, $s=a_+ - a_-$ and the $x_i$'s denote the discrepancies of the fake vertices. So it suffices to prove the following: 
	\begin{claim}
		For each addition of an edge $(v,f)\in (E^+\cup E^-)$ in \Cref{algorithm:prune-fake-recolor}, $\exists\, d'$ such that just after the addition, $\{\disc(f')\,|\,f'\in F\}\subseteq \{d',d'+1\}$. In particular, after the addition of all edges, $\{\disc(f')\,|\,f'\in F\}\subseteq \{\lfloor \frac{a_+ - a_-}{N} \rfloor,\lfloor \frac{a_+ - a_-}{N} \rfloor+1\}$.    
	\end{claim}
	\begin{proof}
		Recall that we first add the edges in $E^+$. Since they are added in round robin fashion, the claim is trivially true up to this point. At this point, there will be some prefix of vertices $F'\subseteq F$ with discrepancy $d'+1$ and the rest have discrepancy $d'$. Now consider the addition of edges in $E^-$. If $|E^-|\leq |F'|$, then nodes in $F\setminus F'$ remain unchanged and the discrepancy of some nodes in $F'$ will become $d'$, thus still satisfying the desired property. If $|E^-| > |F'|$, then after $|F'|$ insertions, all nodes will have discrepancy $d'$, and after this point, discrepancies decrease by $1$ in a round-robin fashion, thus maintaining the desired property. In particular, after the insertion of all edges, we have $d'=\lfloor \frac{a_+ - a_-}{N} \rfloor$. This is because if $\frac{a_+ - a_-}{N}$ is integral, then all vertices in $F$ will have discrepancy $\frac{a_+-a_-}{N}=\lfloor \frac{a_+-a_-}{N} \rfloor$ and if $\frac{a_+ - a_-}{N}$ is non-integral, since it is the average discrepancy, it is a convex combination of $d'$ and $d'+1$, implying $d' = \lfloor \frac{a^- - a^+}{N} \rfloor$. 
	\end{proof}
	As explained in the beginning of the proof, the claim immediately implies that $\Phi(\dirp{H_t}) \leq \Phi(\dir{G}'_{t-1})$.
\end{proof}

\begin{claim}
	\label{cl:potential2}
	For any $j \in \{1, \ldots, N\}$, if $\delta(f_j)$ is the degree of
	$f_j$ in $H_t$, the potential change is
	\[ \Phi(\dirp{H}^{j}) - \Phi(\dir{H}^{j-1}) ~\leq~ O\bigg(
	\frac{\delta(f_j) \; \log m}{\phi^2 \gamma}\bigg). \]
\end{claim}

\begin{proof}
	Let $d$ be the maximum discrepancy of a vertex in
	$\dir{H}^{j-1}$. When we remove the fake vertex $f_j$ from
	$\dir{H}^{j-1}$, the discrepancy of the neighbors of $f_j$ changes
	by 1, and so the potential increases by at most
	$2d \delta(f_j)$. \Cref{lemma:discrepancy-fake-vertices} shows that
	$d$ is $O \big( \frac{ \log m}{\phi^2 \gamma} \big)$.
\end{proof}

\Cref{cl:potential0,cl:potential1,cl:potential2} show that the total increase in the
potential due to deletion of $e_t$, creation of $H_t$ and deletion of a fake vertices is at
most
%
$   \frac{ \log m}{\phi^2 \gamma} + \frac{\vol_0(\Delta P) \log
	m}{\phi^2 \gamma}.$
If $L_t$ denotes the number of flips performed by \local during
\textproc{Prune-and-Recolor}$(G_{t-1}, e_t)$, then 
$$ \Phi(\dir{G}_t) - \Phi(\dir{G}_{t-1}) ~\leq~ \left(\frac{ \log m}{\phi^2 \gamma}  + \frac{\vol_0(\Delta P) \log m}{\phi^2 \gamma}\right) - L_t. $$
This completes the proof of~\Cref{thm:iterated-local-search}.
\end{proof}


We end this section by using~\Cref{thm:iterated-local-search} in an aggregate sense, over a sequence of $t$ adversarial deletions.

\begin{theorem}
	\label{theorem:local-search-recourse}
	Let $G_0 = (V_0,E_0)$ be a $\gamma$-weakly-regular $\phi$-expander
	with $m$ edges and $n$ vertices. Suppose at most
	$D = \nf{\phi^2 m}{20}$ edges are deleted adversarially. Then
	for any $t \leq D$, the total number of edge flips performed by
	\Cref{algorithm:prune-fake-recolor} during the first $t$ deletions
	is at most $O\big(\frac{\log m}{\phi^3 \gamma}\cdot t + m\big)$.
\end{theorem}

\begin{proof}
	The proof is to simply combine \Cref{theorem:pruning-ravi,thm:iterated-local-search} over the sequence of adversarial deletions. Indeed, we can use the facts that the total volume of the pruned set is at most $\vol_0(P_t) \leq 6t/(5\phi)$ along with $\Phi(\dir{G}_0)\leq m$ (optimal offline orientation of $G_0$ has discrepancy at most $n \leq m$), and $\Phi(\dir{G}_t)\geq 0$ to complete the proof.  
\end{proof}

\subsection{Putting Everything Together}
\label{sec:overall}

We now formally describe our overall algorithms and analyses. To keep track of the internal states of the algorithms, we maintain an internal clock which is initialized at \emph{\instant} $\tau=0$ (but eventually $\tau$ will exceed time $t$).
At any
\instant $\tau$, we maintain a decomposition of the current graph
$\currG{\tau}$ into several subgraphs $\{\levelG{i}{\tau}\}_{i \geq 0}$, where  $\levelG{i}{\tau}$ is the {\em level-$i$} subgraph of $\currG{\tau}$. These subgraphs maintain the following invariants: 
\begin{enumerate}  [partopsep=0pt,topsep=0pt,parsep=0pt]
	\item[(I1)] For each \instant $\tau$ and level $i$, the graph $\levelG{i}{\tau}$ has at most $2^i$ edges. 
	\item[(I2)] For every $\tau$ and $i$, subgraph $\levelG{i}{\tau}$ has
	a {\em creation \instant} $\tau_0$ which is at most $\tau$. Graph  $\levelG{i}{\tau}$ is a
	subgraph of $\levelG{i}{\tau_0}$, i.e., we only delete edges from this level between $\tau_0$ and $\tau$. 
	\item[(I3)] For each $\tau$ and $i$, we maintain a decomposition of
	$\levelG{i}{\tau}$ into subgraphs $\expG{i}{j}{\tau}$ for $j \geq 1$, such that any vertex appears in at most $\log^2 n$ of these subgraphs. Moreover, if $\tau_0$ is the creation \instant of $\levelG{i}{\tau},$ then $\expG{i}{j}{\tau_0}$ is $\gamma$-weakly-regular $\phi$-expander for all $j$, and $\expG{i}{j}{\tau}$ is a subgraph of $\expG{i}{j}{\tau_0}$ for all $j$. 
\end{enumerate}

Although not mentioned explicitly in the invariants, the subgraph
$\expG{i}{j}{\tau}$ also has expansion and the weak-regularity
properties given by~\Cref{theorem:pruning-ravi}: in the notation of
this theorem, $\expG{i}{j}{\tau} = G_t$, where $G_0$ is the
corresponding subgraph at the creation \instant $\tau_0$ of
$\levelG{i}{\tau}$.

\medskip\textbf{Edge Insertions.} We first consider the (easier) case of adversarial edge insertions. 
The algorithm appears in \Cref{algorithm:insert}. We first insert the
edge $e$ into level-$1$. Whenever a level-$j$ subgraph overflows
(i.e., has more than $2^j$ edges), we empty this level and move all
the edges to the subsequent level. If this process stops at level $j$,
we build a new expander decomposition of the graph at this level using~\Cref{theorem:gkks20}, and also recompute an optimal offline low-discrepancy discrepancy orientation for each expander. As mentioned before, it is easy to optimally orient any graph in the offline setting: we consistently orient the edges of all cycles, to be left with a forest. We can then again orient all the maximal paths between pairs of leaves in a consistent manner, to end up with an orientation where every vertex has discrepancy in $\{-1,0,1\}$. Note that since it is the optimal discrepancy solution, it is also a locally optimal orientation.

\begin{algorithm}[H]
	\caption{\textproc{Insert}$(e, \tau)$}
	\label{algorithm:insert}
	
	\textbf{Input:} Edge $e$ to be inserted in $\currG{\tau}$.\\
	\textbf{Output:} Graph $\currG{\tau+1}$ with decomposition into levels. 
	
	\begin{algorithmic}[1]
		\State Find the smallest $i$ such that $\{e\}\cup \levelG{1}{\tau}\cup \cdots \cup \levelG{i}{\tau}$ has at most $2^i$ edges. \label{l:i}
		\State Set $\levelG{i'}{\tau+1} = \emptyset$ for $i' =1, \ldots, i-1$. \label{l:empty}
		\State Set $\levelG{i'}{\tau+1} \leftarrow \levelG{i'}{\tau},$ for
		all $i' > i$ and 
		$\levelG{i}{\tau+1} \leftarrow \{e\}\cup \levelG{1}{\tau}\cup \cdots \cup \levelG{i}{\tau}.$
		\State Let the expander decomposition of $\levelG{i}{\tau+1}$ (using \Cref{theorem:gkks20}) 
		return subgraphs $\expG{i}{j}{\tau+1}, j \geq 1$; define their
		\emph{creation \instant} to be $\tau+1$. \label{l:creation}
		\State Find a discrepancy-at-most-$1$ orientation for each $\expG{i}{j}{\tau+1},
		j \geq 1$. \label{l:disc-one}
		\State $\tau\gets \tau+1$ 
	\end{algorithmic}
\end{algorithm}

\medskip\textbf{Edge Deletions.} For the case of adversarial edge
deletions, when an edge $e$ is deleted from subgraph
$\expG{i}{j}{\tau}$, we first check if $\tfrac{\phi^2 m_{i,j}}{20}$
edges have been deleted from $\expG{i}{j}{\tau_0}$, where $\tau_0$ is
the creation \instant of $\levelG{i}{\tau}$ and $m_{i,j}$ was the number
of edges in $\expG{i}{j}{\tau_0}$. If so, we remove the subgraph
$\expG{i}{j}{\tau}$ and re-insert these edges (these are called
\emph{internal inserts}). Otherwise, we run \Cref{algo:prune} on
$\expG{i}{j}{\tau}$ and edge $e$ to get subset $\Delta P$, and then
call~\Cref{algorithm:prune-fake-recolor} which removes $\Delta P$ (via
\emph{secondary deletes}) and reorients edges of
$\expG{i}{j}{\tau}$. Finally, the edges of $\Delta P$ are re-inserted
(causing more \emph{internal inserts}). The algorithm is shown
formally in \Cref{algorithm:delete}.

\begin{algorithm}
	\caption{\textproc{Delete}$(e, \tau)$}
	\label{algorithm:delete}
	\textbf{Input:} Edge $e$ to be deleted from $\currG{\tau}$.\\
	\textbf{Output:} Graph after deletion of edge $e$. 
	\begin{algorithmic}[1]
		\State Find the level $i$ and index $j$ such that $e$ belongs to  $\expG{i}{j}{\tau}.$
		\State Let $\tau_0\gets$ creation \instant of $\levelG{i}{\tau}$, and
		$m_{i,j} \gets$ number of edges in $\expG{i}{j}{\tau_0}$. 
		\State Let $T_{i,j}$ be the set of  $\tau' \in [\tau_0, \tau]$ at which an adversarial edge deletion happened in $\expG{i}{j}{\tau'}$. 
		\If{$|T_{i,j}| \geq \nf{\phi^2 m_{i,j}}{20}$}
		\State Remove all edges in $\expG{i}{j}{\tau}$, re-insert all
		except $e$
		one-by-one using \Cref{algorithm:insert} while incrementing $\tau$. \label{l:ins1}
		\Else
		\State $\Delta P \leftarrow \Prune(\expG{i}{j}{\tau}, e)$ (see
		\Cref{algo:prune}, where $t = |T_{i,j}|$, and $G_{t-1} = \expG{i}{j}{\tau}$).
		\State $G_{i,j}(\tau + 1) \gets \textproc{Prune-and-Reorient}(\expG{i}{j}{\tau},e,
		\Delta P$) (see \Cref{algorithm:prune-fake-recolor}) and $\tau\gets \tau + 1$. \label{l:prunerecolor}
		\State Reinsert edges of $\Delta P$ except $e$ one-by-one while incrementing $\tau$
		using \Cref{algorithm:insert}. \label{l:ins2}
		\EndIf
	\end{algorithmic}
\end{algorithm}


We are now ready to analyze the discrepancy of $\currG{\tau}$ for all
$\tau$, as well as the amortized recourse. We will prove the following quantitative version of \Cref{thm:main2}.

\begin{theorem}[Main Theorem: Graph Orientation]
	\label{theorem:overall}
	Suppose we start with the empty graph on $n$ vertices and it
	undergoes adversarial edge insertions and deletions. There is an algorithm that maintains discrepancy of $O(\log^7 n)$ with an amortized
	recourse of $O(\log^5 n)$ per update.
\end{theorem}

\begin{proof}
    Firstly, we can assume w.l.o.g. that we never have parallel edges, as we can handle repetitions in the following black-box manner. Let $E$ denote the set of active edges and let $E'$ denote the set of active edges in the no-repetitions black-box. For the copies of an edge $e$ in $E$ (call it $T_e$), we will maintain the following invariants: (a) if $|T_e|$ is even, then $e \notin E'$ and half of them are signed $+1$ and, (b) if $|T_e|$ is odd, then $e\in E'$ and the $1$'s and $-1$'s in $T_e$ differ by at most one such that overall the signs add up to $\sigma'(e)$. These invariants ensure that the discrepancy for $E$ is equal to that for $E'$. To maintain these invariants:
    \begin{enumerate}[partopsep=0pt,topsep=0pt,parsep=0pt]
        \item If $|T_e|$ is even and $e$ is added/deleted in $E$, then call insert procedure with $e$ into $E'$ and for each edge $e'$ whose orientation changes in $E'$, you have to flip exactly one copy in $T_{e'}$ to satisfy the invariant.    
        \item If $|T_e|$ is odd and $e$ is added/deleted in $E$, then you have to re-orient at most one of edge in $T_e$ to ensure that $+1$'s and $-1$'s are equal in $T_e$. Then call delete procedure on $e$ from $E'$. Again, for each edge $e'$ flipped in $E'$, you have to flip exactly one copy in $T_{e'}$ to satisfy the invariant.
    \end{enumerate}
	
	For the rest of the proof we assume that there are no parallel edges.
	We first bound the recourse of the algorithm. 
	There are two sources of recourse: 
	\begin{enumerate}  [label=(\alph*),partopsep=0pt,topsep=0pt,parsep=0pt]
		\item While performing a discrepancy-at-most-$1$ orientation
		in~\Cref{l:disc-one} of \textproc{Insert} (\Cref{algorithm:insert}):
		this could happen due to adversarial insert or {\em internal} inserts, i.e., due to lines~\ref{l:ins1} or~\ref{l:ins2} in~\Cref{algorithm:delete}. 
		\item During \local performed inside procedure \textproc{Prune-and-Reorient}, called in line~\ref{l:prunerecolor} of~\Cref{algorithm:delete}. 
	\end{enumerate}

	We first bound the number of calls to the \textproc{Insert} procedure. Let $T$ denote the total number of adversarial inserts and deletes. 
	\begin{claim}
		\label{cl:inserttotal}
		The total number of  calls to \textproc{Insert} procedure is at most $\frac{2T}{\phi^2 \gamma}.$
	\end{claim}
	\begin{proof}
		Clearly, the number of adversarial inserts is at most $T$. First consider the internal inserts caused by line~\ref{l:ins1} of \Cref{algorithm:delete}. These can be charged to the $\phi^2 \gamma m_{i,j}$ adversarial deletes in the set $T_{i,j}$. Therefore, the number of such calls to  \textproc{Insert} procedure is at most $\frac{T}{\phi^2 \gamma}.$ Similarly, the number of inserts in line~\ref{l:ins2} of \Cref{algorithm:delete} is at most $\vol_0(\Delta P)$, the total number of such internal inserts corresponding to a fixed expander graph which undergoes $D'\leq D$ adversarial edge deletions is at most $\frac{6D'}{5 \phi}$ (by~\Cref{theorem:pruning-ravi}). Summing over all expanders, this quantity is at most $\frac{6T}{5 \phi}$. Thus, the overall number of calls to \textproc{Insert}  is at most $\frac{2T}{\phi^2 \gamma}.$
	\end{proof}
	
	We now bound the recourse caused by rebuilding of levels due to insertions. 
	
	\begin{claim}
		\label{cl:recourse1}
		The total re-orientations due to \Cref{l:disc-one} of \textproc{Insert} is at most $\frac{4T \log m}{\phi^2 \gamma}$. 
	\end{claim} 
	\begin{proof}
		For a fixed level $\ist$, let $T'$ be the set of  $\tau$ when we call \textproc{Insert} and the index $i$ selected in line~\ref{l:i} in \Cref{algorithm:insert} happens to be $\ist$. For any such \instant $\tau$, the total number of edges added to  $\levelG{i}{\tau+1}$ is at most 
		$2^\ist$ and at least $2^{\ist -1}$. It follows that the number of re-orientations due to \Cref{l:disc-one} of \textproc{Insert} at this \instant is also at most $2^\ist$.  Note that we empty the levels $1, \ldots, \ist-1$ at this \instant. Therefore, between any two consecutive \instants in $T'$, we must have inserted at least $2^{\ist-1}$ edges (these could be either adversarial or internal inserts).  Thus, the total number of re-orientations due to \Cref{l:disc-one} of \textproc{Insert} for all \instants in $T'$ is at most twice the total number of calls to  \textproc{Insert}, which is at most $\frac{2T}{\phi^2 \gamma}$ (by~\Cref{cl:inserttotal}).

		Since there are at most $\log m$ levels, the desired result follows. 
	\end{proof}
	
	We will next bound the recourse due to local-search steps in the \textproc{Prune-and-Reorient} procedure. 
	\begin{claim}
		\label{cl:recourse2}
		The total number of re-orientations due to \local called in line~\ref{l:localsearchH} of \Cref{algorithm:prune-fake-recolor}  is $O \left( \frac{T \log n}{\gamma \phi^3} \right)$.
	\end{claim}
	\begin{proof}
        Consider any particular expander graph $\expG{i}{j}{\tau_0}$ created at \instant $\tau_0$. Consider the calls to the \textproc{Prune-and-Reorient} procedure where the deleted edge belongs to $\expG{i}{j}{\tau'}$ for some $\tau' \geq \tau_0$.
		We know that $D' \leq D$, this inequality could be strict because we may remove all the level $i$ edges at some \instant because of line~\ref{l:empty} in~\Cref{algorithm:insert}.~\Cref{theorem:local-search-recourse} shows that the total number of re-orientations due to \local called in line~\ref{l:localsearchH} of \textproc{Prune-and-Reorient} procedure during these $D'$ \instants is at most (a constant factor of) 
		$ \frac{D' \log m}{\gamma \phi^3} + m_{i,j}, $
		where $m_{i,j}$ is the number of edges in $\expG{i}{j}{\tau}$. When we add the above for all expanders, the first term is at most 
		$\frac{T \log m}{\gamma \phi^3}$. The second term is the sum over all expanders that get created during the algorithm of the number of edges in the expander. Expanders are created in line~\ref{l:creation} of \Cref{algorithm:insert}, and so their total size can be bounded in the same manner as the argument used in the proof of~\Cref{cl:recourse1}. Hence this quantity is at most $\frac{4T \log m}{\phi^2 \gamma}$.
	\end{proof}
	Since we use $\phi,\gamma = \Theta(1/\log n)$, the above results show that the amortized recourse is 
	at most $O \left( \gamma^{-1}\phi^{-3} \log n \right)= O\left(\log n \cdot \log^3 n \cdot \log n\right) = O(\log^5 n)$. We now bound the discrepancy of any vertex. 
	
	\begin{claim}
		The discrepancy of any vertex is bounded by $O(\log^7 n)$ at all times.
	\end{claim}
	\begin{proof}
		From \Cref{lemma:discrepancy-fake-vertices}, discrepancy of a vertex in any  expander is $O(\phi^{-2}\gamma^{-1}\log m)$, which is $O(\log^4 n)$ since we use $\gamma, \phi = \Theta(1/\log n)$. Since each vertex appears in at most $O(\log^3 n)$ expanders, we get that the discrepancy is bounded by $O(\log^7 n)$ at all times.
	\end{proof}
	
	This completes the proof of \Cref{theorem:overall}.
\end{proof}


\section{Lower Bounds for Local Search}
\label{section-lower-bounds-local-search}

In this section, we will show that for general  vectors and
general graphs, typical $\ell_2$-potential local search procedures do not guarantee low
discrepancy. 

\subsection{The $\ell_2$-Potential for General Vectors}

The $\ell_2$-potential for a signing $\{\e_i\}$ for vectors
$a_i$ is $\sum_j (\sum_i \e_i a_{ij})^2 = \| S \|^2$, where
$S := \sum_i \e_i a_i$. Hence we are at a local optimum if for each
$i$, the potential change $\|S-2\varepsilon_i a_i\|_2 -\|S\|_2$
due to flipping $a_i$ is non-negative. We will show that there exist locally optimal solutions on $T$ vectors with discrepancy $\Omega(\sqrt{T})$.

Consider the following set of vectors in 2 dimensions: $T/2$ vectors of the form $(1,1/\sqrt{T})$ and $T/2$ vectors of the form $(-1,1/\sqrt{T})$ with signing $\varepsilon = \mathbf{1}$. Adding these vectors, we get $S = (0,\sqrt{T})$. Now, for any vector $a$ (w.l.o.g. $a=(1,1/\sqrt{T})$) in the collection, 
\[ 
||S-2a||_2^2 - ||S||_2^2 ~=~ ||(-2,\sqrt{T}-\frac{2}{\sqrt{T}})||_2^2-||(0,\sqrt{T})||_2^2 ~=~ 4 + \left(\sqrt{T}-\frac{2}{\sqrt{T}}\right)^2-T ~=~ \frac{4}{T} ~>~ 0.
\]
Hence, this is indeed a local optimum and has discrepancy $\sqrt{T}$. 

\subsection{The $\ell_2$-Potential for $\{\pm 1\}$-Vectors}

The $\ell_2$-potential for a signing $\{\e_i\}$ for $\pm 1$-vectors
$a_i$ is $\sum_j (\sum_i \e_i a_{ij})^2 = \| S \|^2$, where
$S := \sum_i \e_i a_i$. Hence we are at a local optimum if for each
$i$, the potential change $\|S-2\varepsilon_i a_i\|_2 -\|S\|_2$
due to flipping $a_i$ is non-negative.  
Since $\|a_i\|_2^2 = n$, the above condition is equivalent to showing
\begin{gather}
  S^\intercal (\varepsilon_i a_i) \leq n \label{eq:l2-lopt}
\end{gather}
for all $i$.  We can show
an $\Omega(2^{n/2})$ locality gap in this case.


\begin{lemma}
  \label{lemma:lower-bound-vectors}
  There is a family of instances of $\{\pm 1\}$ vectors, one for each
  $n$ that is a multiple of $8$, having local optima with discrepancy $\Omega(2^{n/2})$ but
  global optima having zero discrepancy.
\end{lemma}
\begin{proof}
  We construct a $\{\pm 1\}$ matrix $M$ with $n$ columns and
  $2\cdot \sum_{i=1}^{n/2} r_i$ rows, where $r_i$ is set later. We prove that
  giving signs $\varepsilon=1$ to the rows of this matrix is a local
  optimum with large discrepancy. Let $S = \pmb{\e}M$ denote the sum
  of the rows of $M$. Our construction consists of \emph{repeating
    units}, where the $i^{\text{th}}$ repeating unit (for $i=1,\ldots ,n/2$) is the following
  $2\times n$ sub-matrix:
  \[
    \begin{bmatrix}
      \begin{bmatrix}
        -1 & -1  \\
        -1 & -1         
      \end{bmatrix} \text{ repeated $i-1$ times},
      \begin{bmatrix}
        1 & 1  \\
        1 & 1         
      \end{bmatrix}, 
      \begin{bmatrix}
        1 & -1  \\
        -1 & 1         
      \end{bmatrix} \text{ repeated $\frac{n}{2} - i$  times}
    \end{bmatrix}_{2\times n} .
  \]
  This unit is repeated $r_i$ times. We will later set $r_i$ to an even number, implying that any vector appears an even number of times. Therefore, by signing these even number of copies in an alternating fashion, we get that the global optimum has discrepancy $0$.  By
  construction, $S = (s_1,s_1,s_2,s_2,...,s_{n/2},s_{n/2})$ for some
  integers $s_j$. Define $\vec{s} := (s_1,s_2, \ldots,s_{n/2})$.
  
  
  \begin{claim}
    \label{claim:pm1-lower-bound}
    Let $B$ be the $\frac{n}{2}\times \frac{n}{2}$ lower-triangular
    matrix with $1$s on the diagonal and $-1$s in the lower
    triangle. Then using
    $\vec{r} := (r_1,\ldots ,r_{n/2})^\intercal = \frac{n}{4}
    (B^{-1})^\intercal B^{-1} \mathbf{1}$ results in $M$
    whose row-sum $S = (s_1,s_1,s_2,s_2,\ldots ,s_{n/2},s_{n/2})$ satisfies
    $\vec{s}:=(s_1,s_2,\ldots ,s_{n/2})^\intercal = 2B^\intercal \vec{r}$. Moreover,
    $\eps_i = 1$ for all $i$ is a local optimum.
  \end{claim}
  \begin{proof}
    The row sum $(s_1,s_1,s_2,s_2,\ldots ,s_{n/2},s_{n/2})$ satisfies:
    \begin{align*}
      2r_1 - \ldots  -2 r_{n/2-1} - 2 r_{n/2}&=s_1\\
      \vdots \\
      2 r_{n/2-1} - 2 r_{n/2} &=s_{n/2-1}\\
      2 r_{n/2} &=s_{n/2} 
    \end{align*}
    which implies $\vec{s} = 2B^\intercal \vec{r}$. Next, we check the
    condition~(\ref{eq:l2-lopt}) for local optimality: for any vector
    $a$ in the $i^{\text{th}}$ repeating unit,  
    \[\langle S, a \rangle = -2s_1 -2s_2 -\ldots  -2s_{i-1} + 2s_i = 2(B\vec{s})_i = 2(B\cdot 2B^\intercal \vec{r})_i = n. \qedhere
    \]
  \end{proof}

  Putting the facts from \Cref{claim:pm1-lower-bound} together, the
  discrepancy vector
  $\vec{s} = 2B^\intercal \vec{r} = \frac{n}{2} B^{-1} \mathbf{1}$. We explicitly write down the 
  inverse of the lower-triangular matrix as follows:
  \[
    \begin{bmatrix}
      1       &           &        &          &       \\
      -1      &  \ddots   &        &          &       \\
      \vdots  &  \ddots   & \ddots &          &       \\
      \vdots  &     & \ddots &  \ddots  &       \\
      -1  &  \cdots   &   \cdots     & -1       &  1    
    \end{bmatrix}_{k\times k}^{-1}
    = 
    \begin{bmatrix}
      1       &           &        &        &       \\
      1       &  \ddots   &        &        &       \\
      2       &  \ddots   & \ddots &        &       \\
      \vdots  &  \ddots   & \ddots & \ddots &        \\
      2^{k-2} &  \cdots   & 2      & 1      &  1       
    \end{bmatrix},
  \]
  Using this, we get that $\vec{s}$ has entries of value
  $\Omega(2^{n/2})$, which proves \Cref{lemma:lower-bound-vectors}.
\end{proof}


\begin{remark}
  Since our example has repetitions, we could try to use this
  structure to assign opposite signs to these multiple copies, and
  thereby get low discrepancy. However, it is easy to extend this to
  avoid repetitions. Since
  $(r_1,\ldots ,r_{n/2}) = \frac{n}{4} (B^{-1})^\intercal
  B^{-1}\mathbf{1}$, the total number of rows
  $R=\sum_i r_i$ in our original example is even and at most
  $2^{4n}$. Take the original matrix $M$, append $2^{4n}-R$ rows of
  $(1,-1,1,-1,\ldots )$ and $(-1,1,-1,1,\ldots )$ in alternation to
  obtain a $2^{4n} \times n$ matrix $M'$. Moreover, append a
  $2^{4n}\times 4n$ matrix of all possible $\{\pm 1\}^{4n}$ vectors to
  the right of $M'$ to get the final $2^{4n} \times 5n$ matrix
  $M''$. The row sum is now $S''=(S,\mathbf{0}^{4n})$, and any vector
  in the first $R$ rows continues to satisfy
  $\langle S'', a \rangle = n \leq 5n$ by construction. For any vector after
  that, $\langle S'', a \rangle = 0 \leq 5n$, because
  $ \langle (s_1,s_1,\ldots ,s_{n/2},s_{n/2}), (1,-1,1,-1,\ldots ) \rangle = 0$. So
  $M''$ has rows in $\{\pm 1\}^{5n}$, and it is a local optimum with
  discrepancy $\Omega(2^{n/2})$.
\end{remark}


\subsection{The $\ell_2$-Potential for General Graphs} \label{sec:local-search-graph-lb}


We saw that the local search with the $\ell_2$ potential was effective
on expanders: however, it fails for general graphs.
We now show instances with $n$ vertices and local optima having
discrepancy $\Omega(n^{1/3})$.
\begin{lemma}
  \label{lemma:lower-bound-graphs}
  There is an infinite family of graph instances with local optima for
  the $\ell_2$-potential having discrepancy $\Omega(n^{1/3})$.
\end{lemma}

\begin{figure}[H]
    \centering
    \includegraphics{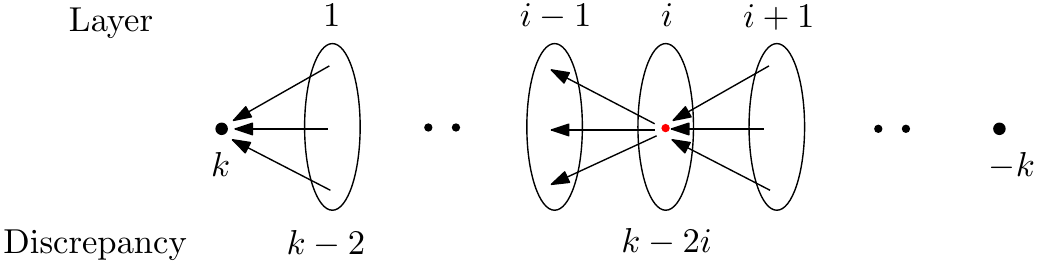}
    \caption{Lower bound of $\Omega(n^{1/3})$ discrepancy for local search on graphs.}
    \label{fig:graphs-lower-bound-simple}
\end{figure}

\begin{proof}
  For even integer $k = \Omega(n^{1/3})$, construct a layered digraph
  with $k$ layers (see \Cref{fig:graphs-lower-bound-simple}). Denote
  the vertices in layer $i$ by $L_i$. For every $u\in L_i$ and
  $v\in L_{i+1}$, add a directed edge from $v$ to $u$.  The number of
  vertices $n_i$ in layer $L_i$ is chosen so that the root has
  discrepancy $k$ and each node in $L_i$ has discrepancy $k-2i$. Since
  a node $u\in L_i$ has incoming edges from every $v\in L_{i+1}$ and
  outgoing edges to every $v'\in L_{i-1}$, it suffices to have
  $n_{i+1}-n_{i-1} = \disc(u) = k-2i$. The base cases are $n_0 = 1$,
  $n_1=k$. Since $k$ is even, this recurrence results in a symmetric
  instance with a zero-discrepancy layer at the center. There are $k$
  layers, and increase in size from $n_{i-1}$ to $n_{i+1}$ is at most
  $k$. So the total number of nodes (up to constant factors) is at
  most $k + 2k + \ldots + (k)k = (1+2+\ldots +k)k = O(k^3)$.

  Finally, each node in $L_i$ has discrepancy equal to $k-2i$ by
  construction. Since all edges (in the directed graph) are of the
  form $u\rightarrow v$ with $u\in L_j, v\in L_{j-1}$ for some $j$, we
  have that $\disc(v)-\disc(u) = (k-2(j-1))-(k-2j) = 2$, and hence the
  orientation is indeed a local optimum.
\end{proof}



\section{Upper-Bounds for Local Search on Unstructured Graphs and Vectors}
\label{sec:local-search-upper}

In \Cref{sec:expand-decomp} we saw how local search with the
$\ell_2$-potential on expander graphs results in small discrepancy;
and \Cref{section-lower-bounds-local-search} gave strong lower bounds
for general graphs, and general collections of vectors. In this
section we give an upper bound for general graphs that matches the
lower bounds; for vectors, we get an exponential-in-$n$ (but
quantitatively weaker) upper bound.

\subsection{Local Search Upper Bound of $n^{O(n)}$ for $\{\pm 1\}$ vectors}
\label{sec:LB-general}

Let us consider a locally optimal configuration. Without loss of generality, $\varepsilon = 1$. For each participating vector $a$, using the fact that $\|a\|_2^2 =n$, we can rewrite the condition $\|S-2a\|_2\geq \|S\|_2$ as $S^T a\leq n$. So for the worst case we are interested in the following program. 
\begin{align*}
    \max_{a_1,\ldots ,a_T\in \{\pm 1\}^n, \,S=\sum_i a_i} \|S\|_\infty\\
    \text{s.t. } \langle S, a_i \rangle \leq n \,\,\forall i\in[T].
\end{align*}

We can show that for $\{\pm 1\}$ vectors, discrepancy at a local optimum is bounded by a function of $n$, independent of $T$. Without loss of generality, $S$ has all positive coordinates. For any vector $u\geq 1$,
    $$\|S\|_\infty \leq \langle S, u \rangle$$
    Suppose there is such a $u$ of the form $\sum x_i a_i$ with $x_i \geq 0$, then we will get $\|S\|_\infty\leq n\cdot(\sum x_i)$ since $\langle S, a_i \rangle \leq n$. Clearly $x_i = 1$ is a feasible solution, but instead let us optimize it as follows.
    \begin{alignat*}{2}
        \min &\textstyle \sum_{i=1}^T x_i &&\\
        x_i &\geq 0 &\qquad\qquad&\text{for $i=1,\ldots ,T$ ($T$ constraints)}\\
        \textstyle \sum_{i=1}^T x_i a_i &\geq 1 &&\text{($n$ constraints)}
    \end{alignat*}
    Let $x^*$ be a corner in the feasible region. By definition of a
    corner, we will have $T$ linearly independent tight constraints at
    $x^*$. Let $n'\leq n$ of them be of the second kind and $T-n'$ of
    the first kind. That is, there are $n'$ non-zero $x^*_i$'s. Then
    we will have $\sum_{i:x^*_i\neq 0} x^*_i a'_i = 1$, where $a'_i$
    is obtained from $a_i$ by retaining only the $n'$ coordinates
    corresponding to tight constraints. Arranging these
    $n'$-dimensional vectors as columns of a matrix, we get a
    $n'\times n'$ matrix $V$ with $V\beta = 1$, where $\beta$ is a
    $n'$ dimensional vector containing the nonzero $x^*_i$'s. $V$ is
    full rank since if the rows had a non-trivial linear combination
    giving zero, then using those coefficients for the $n'$ type-2
    tight constraints will give a vector that is non-zero only at
    positions $i$ where $x^*_i = 0$. So it is a linear combination of
    the $T-n'$ type 1 tight constraints of the form $x_i = 0$. This
    contradicts the linear independence of the tight constraints at a
    vertex. Hence $V$ is full rank, which implies $\beta = V^{-1}\one$. 
    \begin{claim}
    \label{cl:beta}
    The entries of $\beta$ are bounded by $n^{O(n)}$.
    \end{claim}
    \begin{proof}
    For a $k\times k$ matrix $M$, $\det(A)=\sum_{\pi\in S_k}sign(\pi)m_{1,\pi(1)}m_{2,\pi(2)}...m_{k,\pi(k)}$. 
    Since there are at most $k^k$ permutations, any $k\times k$ matrix with $\pm 1$ entries has determinant at most $k^k$. 
    We know that $V^{-1} = \frac{adj(V)}{\det(V)}$ and each entry of $adj(V)$ is the determinant of a $(n'-1)\times (n'-1)$ submatrix of $V$ (i.e., by removing a row and column). 
    Also since $V$ is an invertible $\pm 1$ matrix, $|\det(V)|\geq 1$. Hence we have that entries of $V^{-1}$ are bounded by $(n'-1)^{n'-1} = O(n^n)$, and since $\beta = V^{-1}\one$, the entries of $\beta$ are also bounded by $n^{O(n)}$.
    \end{proof} 
    Now recall that $\|S\|_\infty \leq n\cdot(\sum_{i=1}^T \alpha_i) = n \cdot (\sum_{j=1}^{n'} \beta_j)$ since $\beta$ is a $n'$ dimensional vector containing the non-zero $x_i$'s. \Cref{cl:beta} implies that this quantity is at most $n\cdot(\sum_{j=1}^{n'} n^{O(n)}) \leq n\cdot (n'\cdot n^{O(n)}) = n^{O(n)}$.

\subsection{Local Search for General Graphs: Upper Bounds}
\label{subsection:graphs-L-local-search}

In this section, we will show upper bounds for a simple variant of
\local involving flips along directed paths instead of single
edges. We will refer to it by \pathlocal with parameter $L$
(\Cref{algorithm:path-local}). This is simpler than the method
involving expander decompositions (\Cref{theorem:overall}), but does
not guarantee logarithmic bounds. (In the next section, we show our analysis is tight.) 

\begin{algorithm}
\caption{\pathlocal}
\label{algorithm:path-local}
\textbf{Input:} Graph $G=(V,E)$ and an initial partial coloring, Parameter $L$.\\
\textbf{Output:} Revised orientation which is a local optimum.
\begin{algorithmic}[1]
    \State Arbitrarily direct any undirected edges in $G$.
    \State While there exists a directed path $(u_0,\ldots ,u_l)$ with $l\leq L$ such that $\disc(u_l) > \disc(u_0) + 2$, flip all the edges in the directed path.
\end{algorithmic}
\end{algorithm}  

\begin{figure}
    \centering
    \includegraphics{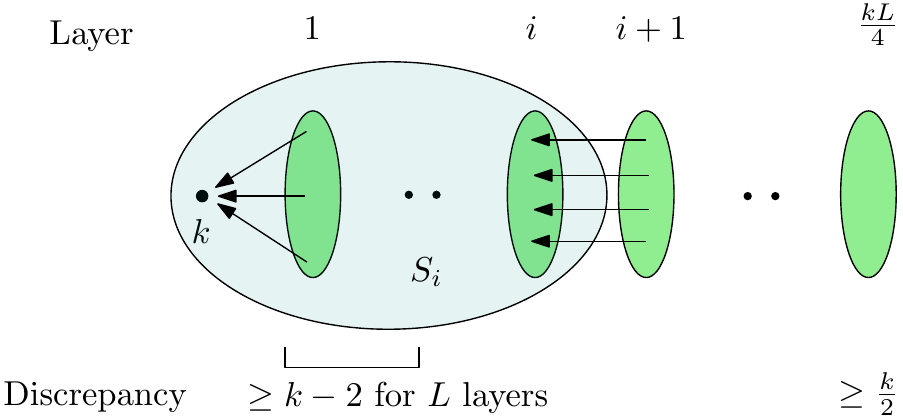}
    \caption{Discrepancy of a local optimum of \pathlocal}
    \label{fig:graphs-upper-bound}
\end{figure}

\begin{theorem}
\label{theorem:local-search-L-tradeoff}
Suppose we start with the empty graph on $n$ vertices and it undergoes adversarial edge insertions and deletions such that at any point, the graph does not have multiple edges. Then there is a deterministic algorithm that achieves $O(D)$ discrepancy with $O(\sqrt{n/D})$ amortized recourse for any $\Omega(1)\leq D \leq O(n)$. 
\end{theorem}

\begin{proof}
Recall that \pathlocal uses the following rule: if there is a directed path $a\rightarrow b$ of length $\leq L$ such that $\disc(b)>\disc(a)+2$, then flip all the edges in the path. Let us bound the discrepancy at a local optimum. 
Let $\vec{G}=(V, \vec{E})$ be the directed graph
corresponding to a local optimum. Consider the node $v$ with
largest discrepancy $k$; without loss of generality, assume
$k \geq 0$. We perform BFS in $\vec{G}$
starting from $v$, following \emph{incoming} edges at each step. Let
$L_i$ be the vertices at level $i$ during this BFS, i.e., $L_i$ is
the set of vertices $w$ for which the shortest path in $\vec{G}$ to
$v$ contains $i$ edges. Let $S_i$ denote the set of vertices in the
first $i$ layers, i.e., $S_i := \bigcup_{i'=0}^i L_{i'}$. 
The fact
that $\vec{G}$ is a local optimum means there are no improving flips,
and hence the discrepancy of any vertex in $L_i$ is at least
$k-2\left(\lfloor \frac{i-1}{L}\rfloor + 1\right)$ (see \Cref{fig:graphs-upper-bound}). 
In turn, this implies that there are at least $kL/2$ layers,
and the discrepancy of any vertex in $S_{kL/4}$ is at least $k/2$.


We now prove the following claim about the rate at which the number of nodes grows.
\begin{claim}
\label{claim:L-local-search-nodes}
Let $n_i$ be the number of nodes in layer $i$ for $0\leq i\leq \frac{kL}{4}$, then $|S_i| = \sum_{j=0}^i n_j \geq \frac{i^2 k}{8}$.
\end{claim}
\begin{proof}
By induction. Base case: Trivial since $|S_0| = 1$.
Induction step: Given a directed graph $\vec{G}$ and a subset $X$ of vertices, let $\delta^-(X)$ denote the set of incoming edges
into $X$ (from $V \setminus X$). Since there are no repeated edges and all edges in $\delta^-(S_i)$ are directed from $L_{i+1}$ to $L_i$, we have $|\delta^-(S_i)| \leq n_i n_{i+1}$.
Since all vertices in $S_i$ have discrepancy at least $k/2$, the number of incoming edges $|\delta^-(S_i)|$ must be at least $\frac{k}{2} \cdot |S_{i}|$. Combining the two inequalities, we have
$n_i n_{i+1}\geq \frac{k}{2} \cdot |S_{i}|.$
Now using this this with the definition of $S_{i+1}$,
\begin{align*}
    |S_{i+1}| ~=~ |S_{i-1}|+n_i+n_{i+1} 
              ~&\geq \frac{(i-1)^2 k}{8} + 2\sqrt{n_i n_{i+1}} \quad \text{(induction hypothesis and AM-GM inequality)}\\
              &\geq \frac{(i-1)^2 k}{8} + 2\sqrt{\frac{k}{2}\cdot|S_i|} \\
              &\geq \frac{(i-1)^2 k}{8} + 2\sqrt{\frac{k}{2}\cdot\frac{i^2 k}{8}} \quad \text{(induction hypothesis)}\\
              &\geq \frac{(i-1)^2 k}{8} + \frac{ik}{2} 
              ~=~ \frac{(i+1)^2 k}{8} \enspace,
\end{align*}
which completes the proof of \Cref{claim:L-local-search-nodes}.
\end{proof}
Now, \Cref{claim:L-local-search-nodes} implies $|S_{kL/4}|\geq (kL/4)^2 k/8 = k^3 L^2/128$. Since we also have $|S_{kL/4}|\leq n$, we get $k\leq O(n^{1/3}/L^{2/3})$. 
To get the amortized recourse, let us track the $l_2$ potential $\Phi = \sum_u \disc(u)^2$, which is initially zero. Before each insertion/deletion, we are at a local optimum, which has a discrepancy of at most $n^{1/3}/L^{2/3}$. So $\Phi$ can change by at most $O(n^{1/3}/L^{2/3})$ when you insert/delete. In each step of local search, $\Phi$ decreases by at least $1$. Since $\Phi \geq 0$, the total number of local search steps is at most $T \cdot \frac{n^{1/3}}{L^{2/3}}$. Since a local search step flips at most $L$ edges, the total recourse is at most $L$ times the number of local search steps. This implies that the amortized recourse is at most $ \frac{n^{1/3}}{L^{2/3}} \cdot L = n^{1/3}L^{1/3}$. This proves \Cref{theorem:local-search-L-tradeoff} in the range $\Omega(1)\leq D\leq O(n^{1/3})$.

For the range $\Omega(n^{1/3}) \leq D \leq n$, we use the following lemma.
\begin{lemma}
\label{lemma:delta-local-search}
Perform local search with $L=1$ but perform a step when $\disc(v)-\disc(u)>\delta$. Then a local optimum has discrepancy at most $n^{1/3}\delta^{2/3}$.
\end{lemma}
\begin{proof}
Very similar to the previous analysis. There will now be $k/2\delta$
layers, which will imply that total number of vertices is at least
$(k/2\delta)^2 k/8 = \Omega(k^3/\delta^2)$. Since this must be at most
$n$, this implies that the discrepancy $k=O(n^{1/3}\delta^{2/3})$. The
potential drops by at least $\delta$ in each local search step, so amortized recourse will be at most $O(n^{1/3}\delta^{2/3}/\delta) = O(n^{1/3}/\delta^{1/3})$. 
\end{proof}
This completes the proof of \Cref{theorem:local-search-L-tradeoff}.
\end{proof}

As an aside, setting $L=1$ and $\delta = 2$ in
\Cref{lemma:delta-local-search} shows the local search in
\Cref{sec:expand-decomp} achieves $O(n^{1/3})$ discrepancy when
performed on general graphs instead of on expanders; this matches the
lower bound in~\Cref{lemma:lower-bound-graphs}.

\subsection{Tightness of discrepancy bound of \pathlocal}

\begin{figure}
    \centering
    \includegraphics{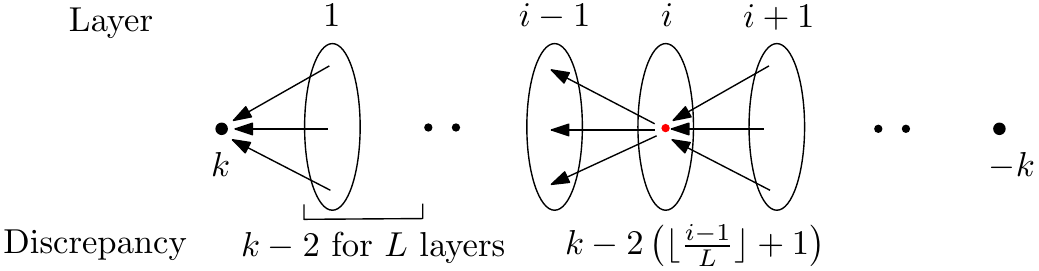}
    \caption{Example showing discrepancy bound at a local optimum of \pathlocal in \Cref{subsection:graphs-L-local-search} is tight.}
    \label{fig:graphs-lower-bound}
\end{figure}

We now provide a family of instances to show that the discrepancy
bound for \pathlocal given in \Cref{subsection:graphs-L-local-search}
is tight. Note that \pathlocal with $L=1$ is the same as \local, and
so using $L=1$ in the lemma below reproduces the $\Omega(n^{1/3})$
lower bound of \Cref{lemma:lower-bound-graphs}.

\begin{lemma}
For \pathlocal with parameter $L$, there is a family of instances of increasing size that are each a local optima and have discrepancy $k$ and number of nodes $O(k^3 L^2)$.
\end{lemma}
\begin{proof}
We will construct a layered directed graph with $O(kL)$ layers (see \Cref{fig:graphs-lower-bound}). The instance will be symmetric and we will use even $k$ and odd $L$. Let us denote layer $i$ by $L_i$. For every $u\in L_i$, $v\in L_{i+1}$, there is a directed edge from $v$ to $u$. We will set up the number of vertices $n_i$ in each layer so that the root has discrepancy $k$ and each node in $L_i$ has discrepancy $k-2\left(\lfloor \frac{i-1}{L}\rfloor + 1\right)$. Let us look at a node $u\in L_i$. It has incoming edges from every $v\in L_{i+1}$ and has outgoing edges to every $v'\in L_{i-1}$. So it suffices to have $n_{i+1}-n_{i-1} = \disc(u) = k-2\left(\lfloor \frac{i-1}{L}\rfloor + 1\right)$. The base cases are $n_0 = 1$, $n_1=k$. It is not hard to see that since $k$ is even and $L$ is odd, this recurrence will result in a symmetric instance with an odd number of zero discrepancy layers in the center. 
There are $O(kL)$ layers and from $n_{i-1}$ to $n_{i+1}$, the increase is at most $k$. So the total number of nodes (up to constant factors) is at most $k + 2k + \ldots  + (kL)k = (1+2+\ldots +kL)k = O(k^3 L^2)$. 
\end{proof}

One can get tightness for the other side of the tradeoff (\Cref{lemma:delta-local-search}) by the same idea.


\subsection{Local Search for Forests}

If we are promised that at each time $t$ the underlying graph is a forest (in the undirected sense), then we can obtain constant discrepancy with logarithmic recourse using the variant of local search that flips along directed paths.
\begin{figure}
    \centering
    \includegraphics{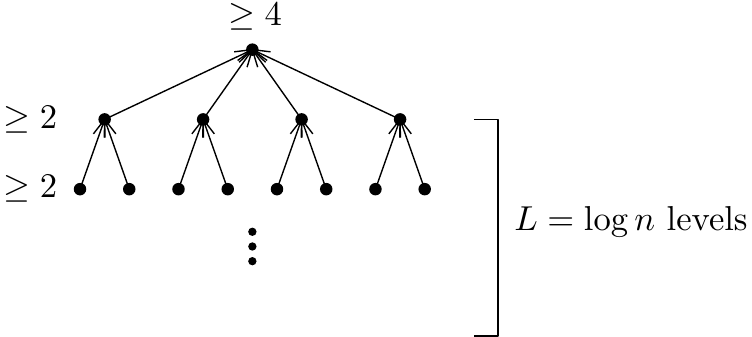}
    \caption{Discrepancy of a local optimum in the forest case}
    \label{fig:trees}
\end{figure}

\begin{lemma}
For forests, performing \pathlocal with $L=\log n$ gives $O(1)$ discrepancy and $O(\log n)$ amortized recourse.
\end{lemma}
\begin{proof}
First, we bound the discrepancy at a local optimum by $3$. For this, suppose for contradiction, the discrepancy (wlog it is positive) is at least 4. Consider the node with discrepancy at least 4 (see \Cref{fig:trees}). It will have a at least 4 incoming edges. Due to local optimum condition with $L=\log n$, each of these children will have discrepancy at least $2$, and therefore will have at least $2$ children. Since $L=\log n$ and forests do not have cycles, we can continue this binary-tree like growth argument till $\log n$ layers. Now  the number of nodes  will be at least $2^{\log n + 1} - 1 = 2n-1$, which is a contradiction. So we get $O(1)$ discrepancy and by the potential based argument as in previous proofs (e.g., \Cref{theorem:local-search-L-tradeoff}), we get $O(\log n)$ amortized recourse. 
\end{proof}



\section{Conclusions and Insert-Only Algorithms} \label{sec:insert-only}

In this paper we initiate the study of  fully dynamic discrepancy problems, where vectors/edges can both arrive
or depart at each time step, and the algorithm must always maintain a low-discrepancy signing. Prior algorithms for online discrepancy could only handle arrivals, and that too only for an oblivious adversary. We obtain near-optimal discrepancy bounds for both the  edge-orientation and the vector balancing cases.  We achieve the former  in near-optimal $\tilde{O}(1)$ amortized recourse, and the latter in $\title{O}(n)$ amortized recourse (which is exponentially better in $T$ than the naive algorithm which recolors after each update).

The main open question left by our work is whether we can achieve near-optimal discrepancy for vector balancing in $\tilde{O}(1)$ amortized recourse per update. If true, this would imply our edge-orientation results as a special case. 

\begin{open} \label{openPb:FullDyn}
For fully-dynamic vector balancing with  vectors of $\ell_2$ length at most $1$ arriving or departing, can we achieve $\poly\!\log(nT)$ discrepancy in $\poly\!\log(nT)$ amortized recourse per update?
\end{open}

Currently,  we don't know how to achieve near-optimal discrepancy  even using $o(n)$ amortized recourse. However, in the case of insertions  only (i.e., no departures),  we can answer the above question affirmatively using the following \DR algorithm. Note that prior works, e.g.~\cite{ALS-STOC21,LSS-arXiv21}, could only achieve this against an oblivious adversary.

The \DR algorithm is simple: when the $t^{th}$ vector arrives,
let $\ell$ be the largest power of $2$ that divides $t$. Use any
offline algorithm $\cA$ to construct a fresh signing of the most
recently arrived $2^\ell$ vectors. 

\begin{theorem}
  For any sequence of $T$ adaptive inserts, suppose that an offline algorithm
  $\cA$ can produce a signing of discrepancy at most $D$ when given
  any subset of these vectors. Then the \DR algorithm achieves 
  at each timestep $t \leq T$ discrepancy  at most
  $D \lceil\log_2 t \rceil$.  Moreover, each vector is assigned a new
  sign at most $\lceil \log_2 T\rceil$ times.
\end{theorem}

\begin{proof}
  Let $t = 2^{a_1} + 2^{a_2} + \ldots + 2^{a_s}$, where
  $a_1 > a_2 > \ldots > a_s \geq 0$. Let
  $\tau_i = \sum_{j \leq i} 2^{a_j}$.  To prove the discrepancy bound,
  the main observation is that at each timestep $t$, the current
  signing consists of the output of $\cA$ on $\lceil \log_2 t \rceil$
  different subintervals of the input sequence:
  $\{v_1, \ldots, v_{\tau_1}\}$, $\{v_{\tau_1+1}, \ldots, v_{\tau_2}\}$,
  all the way down to $\{v_{\tau_{s-1}+1}, \ldots, v_{\tau_s} =
  v_t\}$. (This can be proved using an inductive argument.) The
  discrepancy for each of these logarithmically-many is at most $D$,
  by our assumption on the algorithm $\cA$, which proves the first
  claim. The second claim uses that each time a vector is given a new
  sign, it belongs to an subinterval of twice the length; this can
  happen only $O(\log T)$ times.
\end{proof}

Using the algorithm of~\cite{Banaszczyk-Journal98,BansalDG-FOCS16} gives us the following result.

\begin{corollary} \label{Cor:insertOnly}
  There is an algorithm for the insert-only setting that ensures a
  discrepancy of $O(\sqrt{\log n} \log T)$ for any sequence of vectors
  of $\ell_2$ length at most $1$ (and hence discrepancy of
  $O(\sqrt{s \log n} \log T)$ for any sequence of $s$-sparse vectors   with entries in $[-1,1]$) in  $O(\log T)$ amortized recourse    per update.
\end{corollary}

Another interesting future direction  is to get  near-optimal discrepancy for small \emph{worst-case recourse} per update (instead of amortized recourse). E.g., in the setting of \Cref{openPb:FullDyn},  can we achieve $\tilde{O}(1)$  discrepancy in  $\tilde{O}(1)$ worst-case recourse per update? It will also be interesting to improve \Cref{Cor:insertOnly} to get $O(1)$ amortized recourse per update, or to even get $O(1)$  worst-case recourse per update.

\appendix

\appendix
\section{Missing Details of~\Cref{sec:distrib-BG}}
\label{sec:appsec2}
We first prove~\Cref{lem:dbginvariant}, which is restated below. 

\dbginv*
\begin{proof}
The proof is by induction on the height of $v$: we also add to the induction hypothesis the statement that all the indices $i \in P_v$ such that $-1 < y^v_i < 1,$ belong to $F_v$. For a leaf node, the set $F = P_j$, and so using line~\ref{l:bg}, we see that $\sum_{i \in P_j} y'_i a_i = 0$. Since $y^v_i = y'_i$ for all $i \in P_j$, the invariant~(I1) follows. Invariant~(I2) holds because of~\Cref{lem:BG}. 

Now suppose $v$ is an internal node and assume that the induction hypothesis holds for its children $v_L$ and $v_R$. Since the assignment $x$ just combines $y^{v_L}$ and $y^{v_R}$ (line~\ref{l:combine}), it follows from induction hypothesis that 
$\sum_{i \in P_v} x_i a_i = 0. $
We ensure in line~\ref{l:combine} that $\sum_{i \in F} y'_i a_i = \sum_{i \in F} a_i x_i. $ Therefore, 
$$ \sum_{i \in P_v} y^v_i a_i = \sum_{i \in F} y'_i a_i + \sum_{ i \in P_v \setminus F} x_i a_i = \sum_{i \in P_v} x_i a_i = 0. $$
This proves that~(I1) is satisfied for $y^v$. For property~(I2), first observe that if $i \notin F$, then $x_i \in \{-1, +1\}$ by induction hypothesis, and so $y^v_i = x_i \in \{-1,+1\}$ as well. For the indices $i \in F$, at most $n$ of these satisfy $y'_i \in (-1,+1)$ (by~\Cref{lem:BG}) and so all the variables $y^v_i, i \notin F_v$ are either $+1$ or $-1$. 
\end{proof}

We now give details of the procedure \DBGUpdate in~\Cref{algorithm:recursivebgupdate}. When a vector $a_h$ changes to $a_h^{new}$, we only run the algorithm in~\Cref{lem:BG} for the ancestors of the leaf $j$ in $\calT$ for which $h \in P_j$. We now show that this procedure has the desired properties: 

\begin{claim}
\label{cl:updateBG}
Suppose the assignment $\yold$ satisfies the following properties for every node $v$: (i) $\sum_{i \in P_v} \yold_i a_i = 0$, and (ii) there are at most $n$ indices $i \in P_v$ for which $-1 < \yold_i < +1$. Then the assignment $y^r$, where $r$ is the root node,  returned by $\DBGUpdate(r,\yold,h,a_h^{new})$ also satisfies these properties for every node $v$ (with $a_h$ replaced by $a_h^{new}$). Further $y^r$ and $\yold$ differ in at most $O(n \log T)$ coordinates. 
\end{claim}

\begin{algorithm}[h]
  \caption{Distributed-\bargrin Update: $\DBGUpdate(v, \yold, h, a_h^{new})$}
  \label{algorithm:recursivebgupdate}
  \textbf{Input:} A node $v$ of $\calT$, assignment $\yold$ satisfying~(I1) and (I2), an index $j \in P_v$ where the corresponding vector $a_j$ changes to $a_j^{new}$. \\
  \textbf{Output:} $(y^v, F_v)$: an assignment $y^v_i \in [-1,1]$ for each $i \in P_v$, and $F_v \subseteq P_v$ is the index set of ``fractionally'' signed vectors, i.e., indices $i$ such that $-1 < y^v_i < 1$. 
  \begin{algorithmic}[1]
    \If{$v$ is not a leaf}
    \State Let $v_L$ and $v_R$ be the left and the right children of $v$ respectively. 
    \State Let $\ell \in  \{L,R\}$ be such that $j \in P_{v_\ell}$ and $\ell'$ denote $\{L,R\} \setminus \{\ell\}$. 
    \State $(y^{v_\ell}, F_{v_\ell}) \leftarrow \DBGUpdate(a_\ell, \yold, j, a_j^{new}), $ and $F_{v_{\ell'}} := \{i \in P_{v_{\ell'}} \mid -1 < \yold_i < 1\}. $
    \State Define $F := F_{v_\ell} \cup F_{v_{\ell'}}$, $x_i := y^{v_\ell}_i$ for all $i \in P_{v_\ell}, x_i := \yold_i$ for all $i \in P_{v_{\ell'}}$.  \label{l:combine}
    \Else 
       \State Define $F:= P_v, x_i = 0$ for all $i \in P_v$. 
    \EndIf
    \State Using~\Cref{lem:BG} find a vector $y' \in [-1,1]^{|F|}$ such that (i) $A_F \cdot y'= A_F \cdot  x|_F$, (ii) there are at most $n$ indices, call it $F_v \subseteq F$, such that $-1 < y_i' < 1$ (note that if $h \in F$, then column $h$ of $A_F$ is $a_h^{new}$).  \label{l:combine1}
    \State Define $y^v_i = x_i$ for $i \in P_v \setminus F$ and $y^v_i = y'_i$ for $i \in F$. 
    \State Return $(y^v, F_v)$. 
  \end{algorithmic}
  \end{algorithm}
  
  \begin{proof}
Let the index $h$ belong to $P_j$, where $j$ is a leaf node in $\calT$. Let $w_0 = j, w_1, w_2, \ldots, w_H=r,$ be the path from $j$ to the root $r$ of $\calT$. We prove the following by induction on $\ell$. The assignment $(y^{w_\ell}, F_{v_{\ell}})$ returned by $\DBGUpdate(w_\ell, \yold, h, a_h^{new})$ has the following properties: (i) $\sum_{i \in P_{w_\ell}} y^{w_\ell}_i a_i=0$, (ii) If $-1 < y_i < 1$ for some $i \in P_{w_\ell}$, then $i \in F_{w_\ell}$, (iii) $\yold|_{P_{w_\ell}}$ and $y^{w_\ell}$ differ in at most $2n(\ell+1)$ coordinates. 

The base case when $\ell=0$ follows easily because of line~\ref{l:combine1} and~\Cref{lem:BG}. Now suppose the induction hypothesis is true for $\ell-1$. Assume wlog that $w_{\ell-1}$ is the left child of $w_\ell$ and $w'$ be the right child of $w_\ell$. 
By induction hypothesis and property of $\yold$, we see that (here $x$ is the assignment defined during \DBGUpdate for $w_\ell$):
$$ \sum_{i \in P_{w_\ell}} x_i a_i = \sum_{i \in P_{w_{\ell-1}}} y^{w_{\ell-1}}_i a_i + \sum_{i \in P_{w'}} \yold_i a_i = 0. $$
This proves property~(i). Property~(ii) can be shown similarly. Again, it follows from induction hypothesis and the property of $\yold$ that $|F| \leq 2n$, and so (i) $x$ and $y^{w_\ell}$ differ in at most $2n$ coordinates, and (ii) $\yold|_{P_{w_\ell}}$ and $x$ differ in  at most $2n \ell $ coordinates. This implies property~(iii). 
\end{proof}

\begin{corollary}
\label{cor:BG}
The amortized recourse during a phase of the \DBGUpdate algorithm is $O(n \log N)$. 
\end{corollary}

\begin{proof}
When a phase begins, we run~\Cref{algorithm:recursivebg} to ensure that the assignment $y$ satisfies the conditions stated in~\Cref{cl:updateBG} (for the assignment $\yold$). Using this result, we see that after each update operation, these conditions continue to be satisfied. Therefore, \Cref{cl:updateBG} shows that the recourse encountered after each 
update operation is $O(n \log N)$. 
\end{proof}

{\small
\bibliographystyle{alpha}
\bibliography{bib,refs}

\newcommand{\etalchar}[1]{$^{#1}$}
\begin{thebibliography}{BvdBG{\etalchar{+}}20}

\bibitem[AAG{\etalchar{+}}19]{abboud2019dynamic}
Amir Abboud, Raghavendra Addanki, Fabrizio Grandoni, Debmalya Panigrahi, and
  Barna Saha.
\newblock {Dynamic set cover: improved algorithms and lower bounds}.
\newblock In {\em Proceedings of the 51st Annual ACM SIGACT Symposium on Theory
  of Computing}, pages 114--125, 2019.

\bibitem[AAN{\etalchar{+}}98]{Ajtai98}
Miklos Ajtai, James Aspnes, Moni Naor, Yuval Rabani, Leonard~J Schulman, and
  Orli Waarts.
\newblock Fairness in scheduling.
\newblock {\em Journal of Algorithms}, 29(2):306--357, 1998.

\bibitem[AGZ99]{andrews1999improved}
Matthew Andrews, Michel~X Goemans, and Lisa Zhang.
\newblock {Improved bounds for on-line load balancing}.
\newblock {\em Algorithmica}, 23(4):278--301, 1999.

\bibitem[ALS21]{ALS-STOC21}
Ryan Alweiss, Yang~P. Liu, and Mehtaab Sawhney.
\newblock Discrepancy minimization via a self-balancing walk.
\newblock In {\em Proceedings of STOC}, 2021.

\bibitem[Ban98]{Banaszczyk-Journal98}
Wojciech Banaszczyk.
\newblock {Balancing vectors and Gaussian measures of $n$-dimensional convex
  bodies}.
\newblock {\em Random Struct. Algorithms}, 12(4):351--360, 1998.

\bibitem[Ban10]{Bansal-FOCS10}
Nikhil Bansal.
\newblock {Constructive Algorithms for Discrepancy Minimization}.
\newblock In {\em Proceedings of {FOCS} 2010}, pages 3--10, 2010.

\bibitem[BCH17]{bhattacharya2017deterministic}
Sayan Bhattacharya, Deeparnab Chakrabarty, and Monika Henzinger.
\newblock {Deterministic fully dynamic approximate vertex cover and fractional
  matching in $O(1)$ amortized update time}.
\newblock In {\em International Conference on Integer Programming and
  Combinatorial Optimization}, pages 86--98. Springer, 2017.

\bibitem[BDG16]{BansalDG-FOCS16}
Nikhil Bansal, Daniel Dadush, and Shashwat Garg.
\newblock An algorithm for koml{\'{o}}s conjecture matching banaszczyk's bound.
\newblock In {\em Proceedings of {FOCS} 2016}, pages 788--799, 2016.

\bibitem[BDGL19]{BansalDGL18}
Nikhil Bansal, Daniel Dadush, Shashwat Garg, and Shachar Lovett.
\newblock The {Gram-Schmidt} walk: {A} cure for the {Banaszczyk} blues.
\newblock {\em Theory Comput.}, 15:1--27, 2019.

\bibitem[Bec81]{Beck-Combinatorica81}
J{\'o}zsef Beck.
\newblock {Balanced two-colorings of finite sets in the square I}.
\newblock {\em Combinatorica}, 1(4):327--335, 1981.

\bibitem[BF81]{BeckFiala-DAM81}
J{\'o}zsef Beck and Tibor Fiala.
\newblock {``Integer-making'' theorems}.
\newblock {\em Discrete Appl. Math.}, 3(1):1--8, 1981.

\bibitem[BF99]{MR1734116}
Gerth St\o~lting Brodal and Rolf Fagerberg.
\newblock Dynamic representations of sparse graphs.
\newblock In {\em Algorithms and data structures ({V}ancouver, {BC}, 1999)},
  volume 1663 of {\em Lecture Notes in Comput. Sci.}, pages 342--351. Springer,
  Berlin, 1999.

\bibitem[BG81]{BaranyGrinberg81}
Imre B{\'a}r{\'a}ny and Victor~S Grinberg.
\newblock On some combinatorial questions in finite-dimensional spaces.
\newblock {\em Linear Algebra and its Applications}, 41:1--9, 1981.

\bibitem[BHI18]{bhattacharya2018deterministic}
Sayan Bhattacharya, Monika Henzinger, and Giuseppe~F. Italiano.
\newblock {Deterministic fully dynamic data structures for vertex cover and
  matching}.
\newblock {\em SIAM Journal on Computing}, 47(3):859--887, 2018.

\bibitem[BHN19]{bhattacharya2019new}
Sayan Bhattacharya, Monika Henzinger, and Danupon Nanongkai.
\newblock {A New Deterministic Algorithm for Dynamic Set Cover}.
\newblock In {\em 2019 IEEE 60th Annual Symposium on Foundations of Computer
  Science (FOCS)}, pages 406--423. IEEE, 2019.

\bibitem[BHNW20]{bhattacharya2020improved}
Sayan Bhattacharya, Monika Henzinger, Danupon Nanongkai, and Xiaowei Wu.
\newblock {An Improved Algorithm for Dynamic Set Cover}.
\newblock {\em arXiv preprint arXiv:2002.11171}, 2020.

\bibitem[BJM{\etalchar{+}}21]{BJMSS-SODA21}
Nikhil Bansal, Haotian Jiang, Raghu Meka, Sahil Singla, and Makrand Sinha.
\newblock Online discrepancy minimization for stochastic arrivals.
\newblock In {\em Proceedings of SODA}, pages 2842--2861, 2021.

\bibitem[BJSS20]{BJSS20}
Nikhil Bansal, Haotian Jiang, Sahil Singla, and Makrand Sinha.
\newblock Online vector balancing and geometric discrepancy.
\newblock In {\em Proceedings of STOC}, pages 1139--1152, 2020.

\bibitem[BK19]{bhattacharya2019deterministically}
Sayan Bhattacharya and Janardhan Kulkarni.
\newblock {Deterministically Maintaining a $(2+\epsilon)$-Approximate Minimum
  Vertex Cover in $O(1/\epsilon^2)$ Amortized Update Time}.
\newblock In {\em Proceedings of the Thirtieth Annual ACM-SIAM Symposium on
  Discrete Algorithms}, pages 1872--1885. SIAM, 2019.

\bibitem[BLSZ14]{bosek2014online}
Bartlomiej Bosek, Dariusz Leniowski, Piotr Sankowski, and Anna Zych.
\newblock {Online bipartite matching in offline time}.
\newblock In {\em 2014 IEEE 55th Annual Symposium on Foundations of Computer
  Science}, pages 384--393. IEEE, 2014.

\bibitem[BS20]{BS20}
Nikhil Bansal and Joel~H. Spencer.
\newblock On-line balancing of random inputs.
\newblock {\em Random Struct. Algorithms}, 57(4):879--891, 2020.

\bibitem[BvdBG{\etalchar{+}}20]{bernstein2020fullydynamic}
Aaron Bernstein, Jan van~den Brand, Maximilian~Probst Gutenberg, Danupon
  Nanongkai, Thatchaphol Saranurak, Aaron Sidford, and He~Sun.
\newblock Fully-dynamic graph sparsifiers against an adaptive adversary.
\newblock {\em CoRR}, abs/2004.08432, 2020.

\bibitem[Bá79]{Barany79}
I~Bárány.
\newblock On a class of balancing games.
\newblock {\em Journal of Combinatorial Theory, Series A}, 26(2):115--126,
  1979.

\bibitem[CAHP{\etalchar{+}}19]{cohen2019fully}
Vincent Cohen-Addad, Niklas Oskar~D Hjuler, Nikos Parotsidis, David Saulpic,
  and Chris Schwiegelshohn.
\newblock {Fully Dynamic Consistent Facility Location}.
\newblock In {\em Advances in Neural Information Processing Systems}, pages
  3250--3260, 2019.

\bibitem[CDKL09]{chaudhuri2009online}
Kamalika Chaudhuri, Constantinos Daskalakis, Robert~D. Kleinberg, and Henry
  Lin.
\newblock {Online bipartite perfect matching with augmentations}.
\newblock In {\em IEEE INFOCOM 2009}, pages 1044--1052. IEEE, 2009.

\bibitem[Cha01]{Chazelle-Book01}
Bernard Chazelle.
\newblock {\em {The discrepancy method: randomness and complexity}}.
\newblock Cambridge University Press, 2001.

\bibitem[EL14]{epstein2014robust}
Leah Epstein and Asaf Levin.
\newblock {Robust algorithms for preemptive scheduling}.
\newblock {\em Algorithmica}, 69(1):26--57, 2014.

\bibitem[FW83]{Fagin83}
Ronald Fagin and John~H. Williams.
\newblock A fair carpool scheduling algorithm.
\newblock {\em IBM J. Res. Dev.}, 27(2):133–139, March 1983.

\bibitem[GGK16]{gu2016power}
Albert Gu, Anupam Gupta, and Amit Kumar.
\newblock {The power of deferral: maintaining a constant-competitive Steiner
  tree online}.
\newblock {\em SIAM Journal on Computing}, 45(1):1--28, 2016.

\bibitem[Gia97]{Giannopoulos}
Apostolos~A Giannopoulos.
\newblock On some vector balancing problems.
\newblock {\em Studia Mathematica}, 122(3):225--234, 1997.

\bibitem[GK14]{gupta2014online}
Anupam Gupta and Amit Kumar.
\newblock {Online Steiner tree with deletions}.
\newblock In {\em Proceedings of the twenty-fifth annual ACM-SIAM symposium on
  Discrete algorithms}, pages 455--467. SIAM, 2014.

\bibitem[GKKP17]{Gupta:2017:ODA:3055399.3055493}
Anupam Gupta, Ravishankar Krishnaswamy, Amit Kumar, and Debmalya Panigrahi.
\newblock {Online and Dynamic Algorithms for Set Cover}.
\newblock In {\em Proceedings of the 49th Annual ACM SIGACT Symposium on Theory
  of Computing}, STOC 2017, pages 537--550, New York, NY, USA, 2017. ACM.

\bibitem[GKKS20]{GuptaKKS20}
Anupam Gupta, Ravishankar Krishnaswamy, Amit Kumar, and Sahil Singla.
\newblock Online carpooling using expander decompositions.
\newblock In {\em Proceedings of {FSTTCS}}, pages 23:1--23:14, 2020.

\bibitem[GKKV95]{grove1995online}
Edward~F Grove, Ming-Yang Kao, P.~Krishnan, and Jeffrey~Scott Vitter.
\newblock {Online perfect matching and mobile computing}.
\newblock In {\em Workshop on Algorithms and Data Structures}, pages 194--205.
  Springer, 1995.

\bibitem[GKLX20]{guo2020power}
Xiangyu Guo, Janardhan Kulkarni, Shi Li, and Jiayi Xian.
\newblock {The Power of Recourse: Better Algorithms for Facility Location in
  Online and Dynamic Models}.
\newblock {\em arXiv preprint arXiv:2002.10658}, 2020.

\bibitem[GKS14]{gupta2014maintaining}
Anupam Gupta, Amit Kumar, and Cliff Stein.
\newblock {Maintaining assignments online: Matching, scheduling, and flows}.
\newblock In {\em Proceedings of the twenty-fifth annual ACM-SIAM symposium on
  Discrete algorithms}, pages 468--479. SIAM, 2014.

\bibitem[IW91]{imase1991dynamic}
Makoto Imase and Bernard~M Waxman.
\newblock {Dynamic Steiner tree problem}.
\newblock {\em SIAM Journal on Discrete Mathematics}, 4(3):369--384, 1991.

\bibitem[KKPS14]{MR3238400}
Tsvi Kopelowitz, Robert Krauthgamer, Ely Porat, and Shay Solomon.
\newblock Orienting fully dynamic graphs with worst-case time bounds.
\newblock In {\em Automata, languages, and programming. {P}art {II}}, volume
  8573 of {\em Lecture Notes in Comput. Sci.}, pages 532--543. Springer,
  Heidelberg, 2014.

\bibitem[Kow07]{MR2313458}
\L~ukasz Kowalik.
\newblock Adjacency queries in dynamic sparse graphs.
\newblock {\em Inform. Process. Lett.}, 102(5):191--195, 2007.

\bibitem[LM15]{Lovett-Meka-SICOMP15}
Shachar Lovett and Raghu Meka.
\newblock {Constructive Discrepancy Minimization by Walking on the Edges}.
\newblock {\em {SIAM} J. Comput.}, 44(5):1573--1582, 2015.

\bibitem[{\L}OP{\etalchar{+}}15]{lkacki2015power}
Jakub {\L}{a}cki, Jakub O{\'c}wieja, Marcin Pilipczuk, Piotr Sankowski, and
  Anna Zych.
\newblock {The power of dynamic distance oracles: Efficient dynamic algorithms
  for the Steiner tree}.
\newblock In {\em Proceedings of the forty-seventh annual ACM symposium on
  Theory of computing}, pages 11--20, 2015.

\bibitem[LSS21]{LSS-arXiv21}
Yang~P. Liu, Ashwin Sah, and Mehtaab Sawhney.
\newblock A gaussian fixed point random walk.
\newblock {\em CoRR}, abs/2104.07009, 2021.

\bibitem[Mat09]{Matousek-Book09}
Ji\v{r}\'i Matousek.
\newblock {\em {Geometric discrepancy: An illustrated guide}}, volume~18.
\newblock Springer Science \& Business Media, 2009.

\bibitem[PW93]{phillips1993online}
Steven Phillips and Jeffery Westbrook.
\newblock {Online load balancing and network flow}.
\newblock In {\em Proceedings of the twenty-fifth annual ACM symposium on
  Theory of computing}, pages 402--411, 1993.

\bibitem[Rot14]{Rothvoss14}
Thomas Rothvo{\ss}.
\newblock {Constructive Discrepancy Minimization for Convex Sets}.
\newblock In {\em Proceedings of {FOCS} 2014}, pages 140--145, 2014.

\bibitem[Spe77]{Spencer77}
Joel Spencer.
\newblock Balancing games.
\newblock {\em Journal of Combinatorial Theory, Series B}, 23(1):68--74, 1977.

\bibitem[Spe85]{Spencer85}
Joel Spencer.
\newblock Six standard deviations suffice.
\newblock {\em Trans. Am. Math. Soc.}, 289(2):679--706, 1985.

\bibitem[SSS09]{sanders2009online}
Peter Sanders, Naveen Sivadasan, and Martin Skutella.
\newblock {Online scheduling with bounded migration}.
\newblock {\em Mathematics of Operations Research}, 34(2):481--498, 2009.

\bibitem[ST04]{ST-STOC04}
Daniel~A. Spielman and Shang{-}Hua Teng.
\newblock Nearly-linear time algorithms for graph partitioning, graph
  sparsification, and solving linear systems.
\newblock In {\em Proceedings of the 36th Annual {ACM} Symposium on Theory of
  Computing, Chicago, IL, USA, June 13-16, 2004}, pages 81--90, 2004.

\bibitem[SV10]{skutella2010robust}
Martin Skutella and Jos{\'e} Verschae.
\newblock {A robust PTAS for machine covering and packing}.
\newblock In {\em European Symposium on Algorithms}, pages 36--47. Springer,
  2010.

\bibitem[SW19]{saranurak2019expander}
Thatchaphol Saranurak and Di~Wang.
\newblock Expander decomposition and pruning: Faster, stronger, and simpler.
\newblock In {\em Proceedings of {SODA}}, pages 2616--2635, 2019.

\bibitem[Wes00]{westbrook2000load}
Jeffery Westbrook.
\newblock {Load balancing for response time}.
\newblock {\em Journal of Algorithms}, 35(1):1--16, 2000.

\end{thebibliography}
}

\end{document}